\documentclass{lmcs}

\usepackage[T1]{fontenc}
\usepackage{amsfonts, amssymb, mathtools, stmaryrd}

\usepackage{tikz}
\usepackage{tkz-graph}
\usetikzlibrary{arrows,backgrounds,positioning,fit,cd,babel}
\usetikzlibrary{shapes.geometric}

\usepackage{csquotes}
\usepackage{relsize}
\usepackage{microtype}

\usepackage{hyperref}
\usepackage{thm-restate}
\usepackage{subcaption}
\usepackage{comment}
\usepackage[capitalise,noabbrev]{cleveref}

\usepackage{booktabs}

\DeclareMathOperator{\tw}{tw}

\DeclareMathOperator{\val}{val}

\DeclareMathOperator{\tr}{tr}

\DeclareMathOperator{\pw}{pw}

\DeclareMathOperator{\soe}{soe}

\newcommand{\HomInd}{\textup{\textsc{HomInd}}}

\renewcommand{\phi}{\varphi}
\renewcommand{\epsilon}{\varepsilon}

\newcommand{\sem}[1]{{[\![#1]\!]}}
\newcommand{\Rat}{\mathbb{Q}}
\newcommand{\RatS}{{\mathbb{Q}^S}}
\newcommand{\RatSS}{{\mathbb{Q}^{S \times S}}}
\newcommand{\RatSnS}{{\mathbb{Q}^{S^{n} \times S}}}
\newcommand{\Ac}{{\mathcal{A}}}

\usepackage{xspace}

\newcommand{\AC}{\ensuremath{\textup{AC}^0}\xspace}
\newcommand{\TC}{\ensuremath{\textup{TC}^0}\xspace}
\newcommand{\NL}{\ensuremath{\textup{NL}}\xspace}
\newcommand{\NCone}{\ensuremath{\textup{NC}^1}\xspace}
\newcommand{\NCtwo}{\ensuremath{\textup{NC}^2}\xspace}
\newcommand{\DET}{\ensuremath{\textup{DET}}\xspace}
\newcommand{\LOGSPACE}{\ensuremath{\textup{L}}\xspace}
\newcommand{\CL}{\ensuremath{\textup{C}_=\textup{L}}\xspace}

\newcommand{\SharpL}{\ensuremath{\#\textup{L}}\xspace}
\newcommand{\GapL}{\ensuremath{\textup{GapL}}\xspace}
\newcommand{\PTIME}{\ensuremath{\textup{P}}\xspace}
\newcommand{\coRP}{\ensuremath{\textup{coRP}}\xspace}

\newcommand{\CMSO}{\ensuremath{\mathsf{CMSO}_2}\xspace}

\newcommand\PIT{\textsmaller{PIT}\xspace}
\newcommand\MWA{\textsmaller{MWA}\xspace}
\newcommand\MTA{\textsmaller{MTA}\xspace}

\newcommand\MWAs{\textsmaller{MWA}s\xspace}
\newcommand\MTAs{\textsmaller{MTA}s\xspace}
\newcommand\VCP{\textsmaller{VCP}\xspace}

\definecolor{lipicsGray}{gray}{0.75}
\definecolor{lipicsLightGray}{gray}{0.6}

\tikzset{
	vertex/.style={draw,circle,fill=lipicsGray},
	every node/.style={anchor=center},
	lbl/.style={color=lipicsLightGray}
}

\newcommand\textcite\cite

\theoremstyle{plain}
\newtheorem{theorem}{Theorem}[section]
\newtheorem{lemma}[theorem]{Lemma}
\newtheorem{corollary}[theorem]{Corollary}
\newtheorem{claim}[theorem]{Claim}
\newtheorem{observation}[theorem]{Observation}

\theoremstyle{definition}
\newtheorem{definition}[theorem]{Definition}
\newtheorem{example}[theorem]{Example}

\newenvironment{claimproof}[1][Proof of Claim]{\begin{proof}[#1]}{\end{proof}}

\usepackage{todonotes}

\begin{document}

\title[Homomorphism Indistinguishability, MWA Equivalence, and PIT]{Homomorphism Indistinguishability, Multiplicity Automata Equivalence, and Polynomial Identity Testing}
\titlecomment{This paper is based on extended abstracts presented at \textsmaller{STACS} 2026 \cite{conf_version} and \textsmaller{MFCS} 2024 \cite{seppelt_algorithmic_2024}.}
\author{Marek \v{C}ern\'y \lmcsorcid{0000-0001-6013-2054}}[a]
\author{Tim Seppelt \lmcsorcid{0000-0002-6447-0568}}[b]

\address{Universiteit Antwerpen, Belgium}
\email{marek.cerny@uantwerp.be}

\address{IT-Universitetet i København, Denmark}
\email{tise@itu.dk}

\thanks{We would like to acknowledge fruitful discussions with Miko{\l}aj Boja\'nczyk and Sam van Gool at the Dagstuhl Seminar 25141 \enquote{Categories for Automata and Language Theory}. Furthermore, we are grateful for discussions with David E.\ Roberson, Louis H\"artel, and Georg Schindling. We thank James Main for comments on an earlier draft.
Finally, we gratefully acknowledge Floris Geerts for drawing our attention to the link between automata and Specht--Wiegmann-type theorems. 
Marek \v{C}ern\'y is funded by the University of Antwerp (BOF, Doctoral project 47103). Tim Seppelt is funded by the European Union (CountHom, 101077083). Views and opinions expressed are however those of the author(s) only and do not necessarily reflect those of the European Union or the European Research Council Executive Agency. Neither the European Union nor the granting
authority can be held responsible for them.}


\begin{abstract}
    Two graphs $G$ and $H$ are \emph{homomorphism indistinguishable}
        over a graph class $\mathcal{F}$
        if they admit the same number of homomorphisms from every graph $F \in \mathcal{F}$.
        Many graph isomorphism relaxations such as (quantum) isomorphism and cospectrality can be characterised as homomorphism indistinguishability over specific graph classes.
        Thereby, the problems $\HomInd(\mathcal{F})$ of deciding homomorphism indistinguishability over $\mathcal{F}$ subsume diverse graph isomorphism relaxations whose complexities range from logspace to undecidable.
        Establishing the first general result on the complexity of $\HomInd(\mathcal{F})$, 
        Seppelt (\textsmaller{MFCS}~2024) showed that $\HomInd(\mathcal{F})$ is in randomised polynomial time for every graph class $\mathcal{F}$ of bounded treewidth that can be defined in counting monadic second-order logic $\CMSO$.
        
        We show that this algorithm is conditionally optimal, i.e.\ it cannot be derandomised unless polynomial identity testing is in $\PTIME$.
        For $\CMSO$-definable graph classes $\mathcal{F}$ of bounded pathwidth, 
        we improve the previous complexity upper bound for $\HomInd(\mathcal{F})$ from $\PTIME$ to $\CL$ and show that this is tight. Secondarily, we establish a connection between homomorphism indistinguishability and multiplicity automata equivalence which allows us to pinpoint the complexity of the latter problem as $\CL$-complete. \end{abstract}

\maketitle

\section{Introduction}

    Graph data is ubiquitous: Graphs may represent social networks, transportation networks, chemical or pharmaceutical molecules, databases, or program executions. 
    A central task when presented with graph data is detecting whether two graphs are structurally equivalent or \emph{isomorphic}.
    Graph isomorphism, however, is complexity-theoretically elusive \cite{grohe_graph_2020} and practically not very robust for example w.r.t.\ noise or perturbations. 
    These limitations motivate the study of \emph{graph isomorphism relaxations}, i.e.\ equivalence relations between graphs that are coarser than isomorphism \cite{grohe_thoughts_2025}.
    Although a plethora of graph isomorphism relaxations has been proposed and studied in the past decades, 
    they---to our knowledge---lack a coherent theory which explains e.g.\ their computational complexity.
    In recent years, 
    homomorphism indistinguishability has emerged as a framework that provides increasingly comprehensive answers to this end.

    Two graphs $G$ and $H$ are \emph{homomorphism indistinguishable} over a graph class $\mathcal{F}$ if, for every $F \in \mathcal{F}$, the graphs $G$ and $H$ admit the same number of homomorphisms from the graph $F$.
    Many well-studied graph isomorphism relaxations from a wide range of areas can be characterised as homomorphism indistinguishability relations over natural graph classes.
    Examples include graph isomorphism corresponding to homomorphism indistinguishability over all graphs \cite{lovasz_operations_1967},
    quantum isomorphism (planar graphs, \cite{mancinska_quantum_2020}),
    cospectrality (cycles), and
    equivalence under the $k$-dimensional Weisfeiler--Leman algorithm (graphs of treewidth at most~$k$, \cite{dvorak_recognizing_2010,dell_lovasz_2018}).
    Further examples draw from notions originating in 
    category theory \cite{dawar_lovasz-type_2021,montacute_pebble-relation_2024},
    optimisation \cite{grohe_pebble_2015,atserias_sheraliadams_2013,malkin_sheraliadams_2014,dell_lovasz_2018,grohe_homomorphism_2025,rattan_weisfeiler-leman_2023,roberson_lasserre_2024}, quantum information theory \cite{kar_npa_2025}, machine learning \cite{xu_how_2019,morris_weisfeiler_2019,gai_homomorphism_2025,zhang_beyond_2024}, 
    and finite model theory \cite{grohe_counting_2020,fluck_going_2024,adler_monotonicity_2024,schindling_homomorphism_2025}, see also the monograph \cite{seppelt_homomorphism_2024}.

    The central decision problem associated with homomorphism indistinguishability is $\HomInd(\mathcal{F})$ which asks, given input graphs $G$ and $H$, whether they are homomorphism indistinguishable over a fixed graph class $\mathcal{F}$.
    For varying $\mathcal{F}$, the problems $\HomInd(\mathcal{F})$ subsume diverse graph isomorphism relaxations falling into a wide range of complexity classes.
    For example, for the class of trees $\mathcal{T}$, $\HomInd(\mathcal{T})$ is in polynomial time via the Weisfeiler--Leman algorithm; 
    most notably, for the class of all planar graphs $\mathcal{P}$, $\HomInd(\mathcal{P})$ is undecidable \cite{mancinska_quantum_2020,atserias_quantum_2019,slofstra_set_2019},
    and, for the class of all graphs $\mathcal{G}$,
    $\HomInd(\mathcal{G})$ is graph isomorphism and thus in quasipolynomial time \cite{babai_graph_2016}.
    Note that, despite $\mathcal{T} \subseteq \mathcal{P} \subseteq \mathcal{G}$,
    there is no apparent relation 
    between the complexities of the homomorphism indistinguishability problems over these graph classes, see also \cite[Table 9.1]{seppelt_homomorphism_2024}.
    
    Crucially, by viewing graph isomorphism relaxations as homomorphism indistinguishability relations $\HomInd(\mathcal{F})$, one may analyse their complexity in terms of properties of the graph  classes $\mathcal{F}$, thus facilitating a principled study of the complexity of graph isomorphism relaxations.
    Adopting this approach,
    Seppelt \cite{seppelt_algorithmic_2024} showed that $\HomInd(\mathcal{F})$ is in randomised polynomial time \coRP 
    for every recognisable graph class $\mathcal{F}$ of bounded treewidth. 
    Recognisability is a fairly general property of graph classes introduced by Courcelle \cite{courcelle_monadic_1990} which can be thought of as roughly equivalent to definability in the rather powerful counting monadic second-order logic \CMSO \cite{bojanczyk_definability_2016}.
    By the Robertson--Seymour theorem \cite{robertson_graph_2004}, 
    all minor-closed graph classes are recognisable, the former being central to homomorphism indistinguishability \cite{roberson_oddomorphisms_2022,seppelt_logical_2024}.
    In \cite{seppelt_algorithmic_2024}, it was asked whether the algorithm presented there can be derandomised.
    Our first result asserts that this would be subject to major complexity-theoretic challenges.

    \begin{theorem}[label=thm:recognisable-trees,restate=thmRecogTree]
        There exists a $\CMSO$-definable graph class $\mathcal{F}$ of bounded treewidth
        such that $\HomInd(\mathcal{F})$, \MTA equivalence, and \PIT are logspace many-one interreducible.
    \end{theorem}

    Here, \PIT denotes the polynomial identity testing problem, i.e.\ the problem of deciding whether a polynomial represented by an arithmetic circuit is the zero polynomial. 
    A~deterministic polynomial-time algorithm for \PIT would have far-reaching complexity-theoretic repercussions \cite{kabanets_derandomizing_2003}.
    It was shown in \cite{marusic_complexity_2015} that \PIT is logspace many-one interreducible with equivalence testing for multiplicity tree automata (\MTA). 
    \MTAs \cite{berstel_recognizable_1982} are an algebraic model of computation:
    they assign to every tree a number from a field ($\mathbb{Q}$ in our case).
    One may, for example, think of the trees as representing \textsmaller{XML} documents
    and of the numbers as being probabilities. 
    \MTAs have a wide range of applications, see \cite{marusic_complexity_2015}.
    
    As a secondary result, we streamline the reasoning in \cite{seppelt_algorithmic_2024} by reducing $\HomInd(\mathcal{F})$, for every recognisable graph class $\mathcal{F}$ of bounded treewidth, to \MTA equivalence, see \cref{thm:main1}.
    In retrospect, this is quite natural given Courcelle's automata-theoretic motivation for the notion of recognisability \cite{courcelle_monadic_1990}.
    Analogous to the bounded-treewidth case,
    we show that $\HomInd(\mathcal{F})$ for recognisable graph classes $\mathcal{F}$ of bounded pathwidth reduces to equivalence of multiplicity word automata (\MWA, \cite{schutzenberger_definition_1961}), see \cref{thm:main1}.
    This does not only simplify the reasoning in \cite{seppelt_algorithmic_2024} but also improves the complexity upper bound for such problems  $\HomInd(\mathcal{F})$, via a classical result of Tzeng~\cite{tzeng_path_1996},
    from polynomial time to \DET, the class of problems that are $\NCone$-reducible to computing the determinant of integer matrices \cite{cook_taxonomy_1985}.
    Refining Tzeng's analysis, we establish the precise complexity of \MWA equivalence.
    See \cref{fig:overview} for an overview of complexity classes.

    \begin{figure}
        \centering
        \begin{tikzpicture}[node distance=1.3cm]
            \node (c1) {\LOGSPACE\vphantom{$\textup{L}^{\SharpL}$\CL}};
            \node (c2) [right of=c1] {\NL\vphantom{$\textup{L}^{\SharpL}$\CL}};
            \node (c3) [right of=c2] {\CL\vphantom{$\textup{L}^{\SharpL}$\CL}};
            \node (c4) [right of=c3] {$\textup{L}^{\CL}$\vphantom{$\textup{L}^{\SharpL}$\CL}};
            \node (c5) [right of=c4] {$\textup{L}^{\SharpL}$\vphantom{\CL}};
            \node (c6) [right of=c5] {\DET\vphantom{$\textup{L}^{\SharpL}$\CL}};
            \node (c7) [right of=c6] {\NCtwo\vphantom{$\textup{L}^{\SharpL}$\CL}};
            \node (c8) [right of=c7] {$\textup{NC}$\vphantom{$\textup{L}^{\SharpL}$\CL}};
            \node (c9) [right of=c8] {\PTIME\vphantom{$\textup{L}^{\SharpL}$\CL}};
            \node (c10) [right of=c9] {$\textup{PIT}$\vphantom{$\textup{L}^{\SharpL}$\CL}};
            \node (c11) [right of=c10] {\coRP\vphantom{$\textup{L}^{\SharpL}$\CL}};

            \draw [draw=lipicsLightGray,->] (c1) -- (c2); 
            \draw [draw=lipicsLightGray,->] (c2) -- (c3); 
            \draw [draw=lipicsLightGray,->] (c3) -- (c4); 
            \draw [draw=lipicsLightGray,->] (c4) -- (c5); 
            \draw [draw=lipicsLightGray,->] (c5) -- (c6); 
            \draw [draw=lipicsLightGray,->] (c6) -- (c7); 
            \draw [draw=lipicsLightGray,->] (c7) -- (c8); 
            \draw [draw=lipicsLightGray,->] (c8) -- (c9); 
            \draw [draw=lipicsLightGray,->] (c9) -- (c10); 
            \draw [draw=lipicsLightGray,->] (c10) -- (c11);

            \node [above of=c10, text width=4cm, align=center, anchor=south, yshift=.5cm] (rectw) {$\HomInd(\mathcal{F})$ for bounded-treewidth recognisable $\mathcal{F}$};

            \node [above of=c3, text width=4cm, align=center, anchor=south, yshift=.5cm] (recpw) {$\HomInd(\mathcal{F})$ for bounded-pathwidth recognisable $\mathcal{F}$};

\draw [<->] (c10) edge node [midway, anchor=east] {\cref{thm:recognisable-trees,thm:main1}} (rectw);
            \draw [<->] (c3) edge node [midway, anchor=west] {\cref{thm:pw-rec,thm:pw-hardness}} (recpw);
        \end{tikzpicture}
        \caption{Overview of results relative to some complexity classes. For details on $\textup{L}^{\CL}$, $\textup{L}^{\SharpL}$, and \DET, see \cite{allender_complexity_1999,allender_arithmetic_2004,mahajan_determinant_1997}.
        $\textup{PIT}$ denotes the class of problems logspace many-one reducible to \PIT.}
        \label{fig:overview}
    \end{figure}

    \begin{theorem}[restate=thmMarek,label=thm:main3]
        \MWA equivalence is \CL-complete under logspace many-one reductions.
    \end{theorem}

    Combined with the initially described reduction,
    this yields the following upper bound for the complexity of $\HomInd(\mathcal{F})$ for recognisable bounded-pathwidth graph classes $\mathcal{F}$.

    \begin{theorem}\label{thm:pw-rec}
        For every recognisable graph class $\mathcal{F}$ of bounded pathwidth,
        $\HomInd(\mathcal{F})\! \in~\!\!\CL$.
    \end{theorem}

    Since \CL is a subclass of $\NCtwo$,
    \cref{thm:pw-rec} implies that homomorphism indistinguishability over every recognisable bounded-pathwidth graph class can be decided on a parallel computer with polynomially many processors in time $O(\log^2n)$.
    Thereby, the theorem establishes a trade-off between distinguishing power and computational complexity. 
    While homomorphism indistinguishability over bounded-pathwidth graph classes is provably weaker than homomorphism indistinguishability over bounded-treewidth graph classes \cite[Theorem~6.4.6]{seppelt_homomorphism_2024}, it can be decided at a substantially lower computational cost.
    This is of particular interest for the graph learning community where message-passing graph neural networks are an often-employed paradigm, see \cite{grohe_logic_2021}. 
    With respect to distinguishing power, these architectures correspond to the Weisfeiler--Leman algorithm \cite{xu_how_2019,morris_weisfeiler_2019} which cannot be efficiently parallelised unless $\PTIME = \textup{NC}$ \cite{grohe_equivalence_1999}.
    In contrast, our \cref{thm:pw-rec} provides a wealth of efficiently parallisable graph isomorphism relaxations.

    In order to get an intuition for the complexity class \CL as introduced in \cite{allender_relationships_1996},
    one may consider its complete problems such as the set of singular integer matrices \cite{allender_arithmetic_2004}
    or the machine characterisation \cite{allender_complexity_1999,allender_relationships_1996} asserting that a language $A \subseteq \{0,1\}^*$ is in \CL if, and only if, there exists a non-deterministic logspace Turing machine $M$ such that $x \in A$ iff the number of accepting paths equals the number of rejecting paths of $M$ on input $x$.
    Many linear-algebraic problems were shown to be complete for \CL \cite{hoang_complexity_2000,mahajan_complexity_2010}, see also \cite{allender_arithmetic_2004,hoang_complexity_2003,mahajan_determinant_1997}.
    We add to the list of such problems by showing that the upper bound in \cref{thm:pw-rec} is tight.

    \begin{theorem}[restate=thmRecogPath, label=thm:pw-hardness]
        There exists a \CMSO-definable graph class $\mathcal{F}$ of bounded pathwidth
        such that $\HomInd(\mathcal{F})$ is $\CL$-complete under logspace many-one reductions.
    \end{theorem}

    \subsection*{Related Work}
    
    Our results provide a complexity-theoretic foundation 
    for the many-faceted connections between systems of equations and homomorphism indistinguishability. 
    Two graphs $G$ and $H$ are isomorphic if, and only if, there exists a permutation matrix $X$ such that $X A_G = A_H X$ where $A_G$ and $A_H$ denote the adjacency matrices of $G$ and $H$, respectively.
    When relaxing the constraints on the matrix $X$, one obtains characterisations of other homomorphism indistinguishability relations.
    For example, two graphs $G$ and $H$ are homomorphism indistinguishable over all paths iff there exists a pseudo-stochastic matrix $X$ satisfying the above equation \cite{dell_lovasz_2018}.
    Thus, homomorphism indistinguishability of simple graphs generalises a variety of relaxations of permutation-similarity for symmetric $\{0,1\}$-matrices. 
    Various relaxations of similarity for integer matrices have been shown \cite{hoang_complexity_2000,hoang_complexity_2003} to be complete for $\CL$ and $\LOGSPACE^{\CL}$. 
    Less is known about symmetric $\{0,1\}$-matrices, see \cref{sec:undirected}.
    
    The fixed-parameter complexity with respect to logarithmic space of homomorphism counting was studied in \cite{chen_fine_2015,haak_parameterised_2023}.
    However, this task is only tangentially related to homomorphism indistinguishability.
    Since the considered graph classes $\mathcal{F}$ are typically infinite,
    being able to count homomorphisms from $F \in \mathcal{F}$ does not a priori help to decide $\HomInd(\mathcal{F})$.
    In fact, the algorithms in \cref{thm:pw-rec} and \cite{seppelt_algorithmic_2024} compute linear-algebraic invariants rather than homomorphism counts.
    Conversely, it is not clear how to infer homomorphism counts using $\HomInd(\mathcal{F})$.

    Finally, 
    we would like to highlight two previous results on the complexity of homomorphism indistinguishability within polynomial time (both results are originally proven in model-theoretic terms but can be recast via \cite{dvorak_recognizing_2010,dell_lovasz_2018,grohe_counting_2020}):
    Firstly, Grohe~\cite{grohe_equivalence_1999} showed that, for every $k \geq 1$,
    homomorphism indistinguishability over all graphs of treewidth at most $k$ is complete for \PTIME under uniform \AC-reductions.
    Note that this does not have implications for other graph classes of bounded treewidth, e.g.\ the class of outerplanar graphs or the class in \cref{thm:recognisable-trees}.
    Secondly, Grohe and Verbitsky \cite{hutchison_testing_2006} 
    showed that, for every $k \geq 1$,
    homomorphism indistinguishability over all graphs of treedepth at most $k$
    can be decided in~\TC, see also~\cite{rasmann_finite_2025,grohe_counting_2020}.
    Again, this does not have implications for other bounded-treedepth graph classes.
    Since bounded treedepth implies bounded pathwidth \cite{nesetril_tree-depth_2006},
    \Cref{thm:pw-rec} applies to a much wider class of graph classes while yielding containment in \CL rather than \TC.

    \section{Preliminaries}

    The natural numbers are $\mathbb{N} = \{0,1,2,\dots\}$.
    In all computational tasks, rationals are encoded as fractions of binary integers.

    \subsection{Multiplicity automata}
    A~\emph{multiplicity word automaton} (\MWA, \cite{schutzenberger_definition_1961}) is a tuple $\Ac = (S, \Sigma, M, \alpha, \eta)$,
    where $S$ is a~finite set of states, $\Sigma$ is a~finite alphabet,
    $\alpha^\top \in \RatS$ is the \emph{initial vector}, and $\eta \in \RatS$ is the \emph{final vector}.
    The map $M \colon \Sigma \to \RatSS$ assigns to each letter $a \in \Sigma$ a~\emph{transition matrix}.
    For convenience, we extend $M$ from $\Sigma$ to all words $w = a_1 \cdots a_t \in \Sigma^*$ by defining
    $
        M(w) \coloneqq M(a_1) \cdots M(a_t), 
    $
    so that the empty word $\epsilon$ is mapped to the identity matrix $M(\epsilon) = I_S$.
    The automaton $\Ac$ \emph{recognises} the rational series $\sem{\Ac}\colon \Sigma^* \to \Rat$ given by $w \mapsto  \alpha^\top M(w) \eta$.

    A~\emph{finitary type} is a set $\Omega$ of symbols with 
    an~\emph{arity} map assigning to each symbol $\sigma \in \Omega$ a~natural number $|\sigma|$.
    For $n\in \mathbb N$, the set of all $n$-ary symbols in $\Omega$ is denoted by $\Omega_n$.
    The set of \emph{$\Omega$-trees} denoted by $T_\Omega$ is the smallest set such that $\Omega_0 \subseteq T_\Omega$,
    and if $n\ge 1$, symbol $\sigma\in \Omega_n$ and $t_1, \dots, t_n \in T_\Omega$, then the element $\sigma(t_1, \dots, t_n)$ is in $T_\Omega$.
    
    The \emph{Kronecker product} $M\otimes M'\in \Rat^{(S\times S') \times (R\times R')}$ 
    of matrices $M\!\in\Rat^{S \times R}$ and $M'\!~\!\in~\!\Rat^{S'\times R'}$
    is given by $(s s', rr') \mapsto M(s, r)\cdot M'(s',r').$
    The \emph{direct sum} $M\oplus M'\in \Rat^{(S\uplus S')\times (R\uplus R')}$ is given by the original entries of $M$ and $M'$, assigning $0$ to the remaining entries in $(S\times R') \uplus (S'\times R)$.
    
    A \emph{multiplicity tree automaton} (\MTA, \cite{berstel_recognizable_1982}) is a~tuple $\Ac = (S, \Omega, \mu, \eta)$,
    where $S$ is a finite set of states, $\Omega$ is a~finitary type, 
    $\mu$ is 
    a \emph{tree representation}, i.e.\ a union of maps $M_n\colon \Omega_n \to \RatSnS$ for each arity $n$,
    and $\eta\in \RatS$ is the \emph{final vector}.
    For each symbol $\sigma \in \Omega_n$, the matrix $\mu(\sigma) = M_n(\sigma) \in \RatSnS$ is called the \emph{transition matrix}.
    We extend the tree representation $\mu$ from $\Omega$ to all elements $\sigma(t_1, \dots, t_n) \in T_\Omega$ by defining
    \begin{align*}
        \mu(\sigma(t_1, \dots, t_n)) &\coloneqq \left(\mu(t_1) \otimes \cdots \otimes \mu(t_n)\right) \cdot \mu(\sigma). 
    \end{align*}
    
    The automaton $\Ac$ \emph{recognises} the series $\sem{\Ac}\colon T_\Omega \to \Rat$ given by
    $t \mapsto \mu(t) \cdot \eta$.
    Two \MTAs $\Ac$ and $\Ac'$ over $\Omega$ are \emph{equivalent} if $\sem{\Ac} = \sem{\Ac'}$.
    The decision problem of \MTA equivalence~\cite{marusic_complexity_2015} assumes rational entries given as fractions of binary integers.
    Note that \MWAs are a special case of \MTAs via $(S, \Omega_0 \cup \Sigma, M_0 \cup M, \eta)$,
    where $\Omega_0 = \{\sigma_0\}$ and $M_0\colon\{\sigma_0\} \to \{\alpha\}$.

    The following operations on rational word series correspond to operations on \MWAs.
    Let $\Ac' = (S', \Sigma, M', \alpha', \eta')$ be another \MWA over the same alphabet $\Sigma$.
    \begin{description}
        \item[Sum lifts to direct sum] The \MWA $\Ac \oplus \Ac'$ with states $S \uplus S'$ is given by the initial vector $\alpha\oplus \alpha'$,
        transition matrix $M(a)\oplus M'(a)$ for each $a\in \Sigma$ and the final vector $\eta\oplus \eta'$.
        The \MWA $\Ac \ominus \Ac'$ is given analogously except for the final vector $\eta\oplus (-\eta')$.
        It holds that $\sem{\Ac \oplus \Ac'}(w) = \sem{\Ac}(w) + \sem{\Ac'}(w)$ and $\sem{\Ac \ominus \Ac'}(w) = \sem{\Ac}(w) - \sem{\Ac'}(w)$ for each $w \in \Sigma^*$.
        \item[Zero series lifts to the zero automaton] Assume $\Ac'$ has no states $S' = \emptyset$, hence the underlying vector space is of dimension 0. Then $\Ac'$ recognizes the zero series, that is $\sem{\Ac'}(w) = 0$ for each word $w \in \Sigma^*$. We call $\Ac'$ the \emph{zero automaton} over $\Sigma$.
        \item[Product lifts to Kronecker product] The \MWA $\Ac \otimes \Ac'$ with states $S \times S'$ is given by the initial vector $\alpha\otimes \alpha'$, transition matrix $M(a)\otimes M'(a)$ for each $a\in \Sigma$ and the final vector $\eta\otimes \eta'$.
        It holds that $\sem{\Ac \otimes \Ac'}(w) = \sem{\Ac}(w) \cdot \sem{\Ac'}(w)$ for each word $w \in \Sigma^*$.
    \end{description}
     These operations naturally extend to \MTAs \cite{berstel_recognizable_1982, marusic_complexity_2015}.

	\subsection{Relational structures and logic}
	
	We assume familiarity with standard notions from finite model theory, see e.g.\ \cite{immerman_descriptive_1999}.
	All signatures, structures, and graphs are finite.
	Let $\tau$ denote a relational signature.
	A class $\mathcal{C}$ of $\tau$-structures is \emph{$\mathsf{L}$-definable} for some logic $\mathsf{L}$ if there exists a sentence $\phi \in \mathsf{L}$ such that a $\tau$-structure $A$ satisfies $\phi$ if, and only if, $A \in \mathcal{C}$.
    For some implicit relational signature $\tau$, we write $\mathsf{MSO}$ for monadic second-order logic over $\tau$.
    There are several options for encoding graphs as relational structures in the context of $\mathsf{MSO}$.
    The less powerful variant $\mathsf{MSO}_1$ encodes a graph $(V, E)$ as relational structure with universe $V$ and the edge relation.
    In the more powerful variant, $\mathsf{MSO}_2$ the graph is encoded as relational structure with universe $V \uplus E$ and the incidence relation.
    See \cite{courcelle_graph_2012} for further details.

    Let $A$ and $B$ be $\tau$-structures.
    A \emph{homomorphism} $h \colon A \to B$ is a map from the universe of $A$ to the universe of $B$ such that $h(R^A) \subseteq R^B$ for all relation symbols $R \in \tau$.
    For example, a homomorphism $h \colon F \to G$ between simple graphs $F$ and $G$ is a map $V(F) \to V(G)$ such that $h(uv) \in E(G)$ for all $uv \in E(F)$.
    We write $\hom(A, B)$ for the number of homomorphisms from $A$ to $B$.
    Two $\tau$-structures $A$ and $B$ are \emph{homomorphism indistinguishable} over a class of $\tau$-structures $\mathcal{F}$, in symbols $A \equiv_{\mathcal{F}} B$, if $\hom(F, A) = \hom(F, B)$ for all $F \in \mathcal{F}$.
    We will ultimately be interested in homomorphism indistinguishability of simple graphs over classes of simple graphs, which has been the main focus in the literature, see \cite[p.\ 32]{seppelt_homomorphism_2024}.

    \subsection{Logarithmic space and related complexity classes}

    Let \LOGSPACE denote the class of languages decided by a deterministic logarithmic-space Turing machine.
    See \cite{arora_computational_2009} for a definition of logspace many-one reductions.
    Let \SharpL denote the class of functions $f \colon \{0,1\}^* \to \mathbb{N}$ that count the number of accepting paths of a non-deterministic logarithmic-space Turing machine \cite{alvarez_very_1993}.
    Let \GapL denote the class of functions of the form $f - g$ for some functions $f, g \in \SharpL$ \cite{allender_relationships_1996}.
    Finally, let \CL denote the class of languages of the form $\{x \in \{0,1\}^* \mid f(x) = 0\}$ for some function $f \in \GapL$ \cite{allender_relationships_1996}.
    It was shown by \textcite{damm_det_1991,vinay_counting_1991,toda_counting_1984} 
    that \GapL coincides with the class of functions which are logspace many-one reducible to the determinant, see \cite{allender_complexity_1999,mahajan_determinant_1997}.
    A complete problem for $\CL$ is the set of singular integer matrices, cf.\ \cite{allender_arithmetic_2004}.
    To get a feeling for \CL, we make the following observation in full detail.
    
    \begin{observation}\label{obs:paths}
        Homomorphism indistinguishability over directed cycles is in \CL.
    \end{observation}
    \begin{proof}
        Slightly abusing notation, we write $\vec{C}_0$ for the edge-less one-vertex graph.
        The cycle $\vec{C}_k$ for $k \geq 1$ has $k$ vertices and $k$ directed edges, i.e.\ $\vec{C}_1$ is a single vertex with a loop.
        
        First, note that the function $f \colon (G, k) \mapsto \hom(\vec{C}_k, G)$ is in $\SharpL$.
        Here, $G$ is a directed graph and $k \geq 0$.
        A non-deterministic logspace Turing machine whose number of accepting paths is $\hom(\vec{C}_k, G)$
        operates by guessing a vertex $v_1 \in V(G)$
        and, for $2 \leq i \leq k$, guesses a vertex $v_i \in V(G)$ and rejects if $v_{i-1}v_i \not\in E(G)$.
        Finally, it rejects if $v_k v_1 \not\in E(G)$ and accepts otherwise.
        Clearly, the machine accepts iff $v_1 \dots v_k$ is the homomorphic image of $\vec{C}_k$
        Since it suffices to store indices of three vertices, i.e.\ $v_1$, $v_{i-1}$, and $v_i$, 
        the space requirement is $3 \log(|V(G)|) + \log(k)$.

        As the difference of two $\SharpL$-functions \cite[Proposition~2]{allender_relationships_1996},
        the function $p \colon (G, H, k) \mapsto \hom(\vec{C}_k, G) - \hom(\vec{C}_k, H)$ is in $\GapL$.
        By \cite[Theorem~9]{allender_relationships_1996},
        the function $q \colon (G, H) \mapsto \sum_{k=0}^{n} p(G,H,k)^2$ is in $\GapL$ where $n \coloneqq \max\{ |V(G)|, |V(H)|\}$.
        It holds that $q(G, H) = 0$ if, and only if, $\hom(\vec{C}_k, G) = \hom(\vec{C}_k, H)$ for all $0 \leq k \leq n$.
        By Newton's identities \cite[2.4.P10]{horn_matrix_2010}, the latter holds if, and only if, $G$ and $H$ are homomorphism indistinguishable over  directed cycles $\vec{C}_k$ for arbitrary length $k \geq 0$.
        Thus, the claim follows by the definition of $\CL$ \cite[Definition~2]{allender_relationships_1996}.
    \end{proof}

    \subsection{Labelled graphs and homomorphism tensors}

    We recall some of the notation from \cite{seppelt_algorithmic_2024}.
    Let $k \geq 1$. 
	A \emph{distinctly $k$-labelled graph} is a tuple $\boldsymbol{F} = (F, \boldsymbol{u})$ where $F$ is a graph and $\boldsymbol{u} \in V(F)^k$ is such that $u_i \neq u_j$ for all $1 \leq i < j \leq k$.
	We say $u_i \in V(F)$, the $i$-th entry of $\boldsymbol{u}$, \emph{carries the $i$-th label}.
	Write $\mathcal{D}(k)$ for the class of distinctly $k$-labelled graphs.

    Let $k, \ell \geq 1$.
	A \emph{distinctly $(k,\ell)$-bilabelled graph} is a tuple $\boldsymbol{F} = (F, \boldsymbol{u}, \boldsymbol{v})$ where $F$ is a graph and $\boldsymbol{u} \in V(F)^k$ and $\boldsymbol{v} \in V(F)^\ell$ are such that $u_i \neq u_j$ for all $1 \leq i < j \leq k$ and $v_i \neq v_j$ for all $1 \leq i < j \leq \ell$.
	Note that $\boldsymbol{u}$ and $\boldsymbol{v}$ might share entries.
	We say $u_i \in V(F)$ and $v_i \in V(F)$ \emph{carry the $i$-th in-label} and \emph{out-label}, respectively.
	Write $\mathcal{D}(k,\ell)$ for the class of distinctly $(k,\ell)$-bilabelled graphs.
	
	For a graph $G$, and $\boldsymbol{F} = (F, \boldsymbol{u}) \in \mathcal{D}(k)$ define the \emph{homomorphism tensor}~$\boldsymbol{F}_G \in \mathbb{N}^{V(G)^k}$ of $\boldsymbol{F}$ w.r.t.\ $G$ 
	whose $\boldsymbol{v}$-th entry is equal to the number of homomorphisms $h \colon F \to G$ such that $h(u_i) = v_i$ for all $i \in [k]$.
	Analogously, for $\boldsymbol{F} \in \mathcal{D}(k,\ell)$, define $\boldsymbol{F}_G \in \mathbb{N}^{V(G)^k \times V(G)^\ell}$.
	
	As observed in \cite{mancinska_quantum_2020,grohe_homomorphism_2025}, 
	(bi)labelled graphs and their homomorphism tensors are intriguing due to the following correspondences between combinatorial operations on the former and algebraic operations on the latter: \label{operations}
	\begin{description}
		\item[Dropping labels corresponds to sum-of-entries.] For $\boldsymbol{F} = (F, \boldsymbol{u}) \in \mathcal{D}(k)$, define $\soe(\boldsymbol{F}) \allowbreak \coloneqq F$ as the underlying unlabelled graph of~$\boldsymbol{F}$.
		Then for all graphs $G$, $\hom(\soe (\boldsymbol{F}), G) = \sum_{\boldsymbol{v} \in V(G)^k} \boldsymbol{F}_G(\boldsymbol{v}) \eqqcolon \soe(\boldsymbol{F}_G)$.
		\item[Gluing corresponds to Schur products.]
		For $\boldsymbol{F} = (F, \boldsymbol{u})$ and $\boldsymbol{F}' = (F', \boldsymbol{u}')$ in $\mathcal{D}(k)$,
		define $\boldsymbol{F} \odot \boldsymbol{F}' \in \mathcal{D}(k)$ as the $k$-labelled graph obtained by taking the disjoint union of $F$ and $F'$ and placing the $i$-th label at the vertex obtained by merging $u_i$ with $u'_i$ for all $i \in [k]$.
		Then for every graph~$G$ and $\boldsymbol{v} \in V(G)^k$, $(\boldsymbol{F} \odot \boldsymbol{F}')_G(\boldsymbol{v}) = \boldsymbol{F}_G(\boldsymbol{v}) \boldsymbol{F}'_G(\boldsymbol{v}) \eqqcolon (\boldsymbol{F}_G \odot \boldsymbol{F}'_G)(\boldsymbol{v})$.
		One may similarly define the gluing product of two $(k,\ell)$-bilabelled graphs.
		\item[Series composition corresponds to matrix products.]
		For bilabelled graphs $\boldsymbol{K} = (K, \boldsymbol{u}, \boldsymbol{v})$ and $\boldsymbol{K}' = (K', \boldsymbol{u}', \boldsymbol{v}')$ in $\mathcal{D}(k,k)$,
		define $\boldsymbol{K} \cdot \boldsymbol{K}' \in \mathcal{D}(k,k)$
		as the bilabelled graph obtained by taking the disjoint union of $K$ and $K'$,
		merging the vertices $v_i$ and $u'_i$ for $i \in [k]$, 
		and placing the $i$-th in-label (out-label) on $u_i$ (on $v'_i$) for $i \in [k]$.
		Then for all graphs $G$ and $\boldsymbol{x}, \boldsymbol{z} \in V(G)^k$,
		\(
		(\boldsymbol{K} \cdot \boldsymbol{K}')_G(\boldsymbol{x}, \boldsymbol{z})
		=\sum_{\boldsymbol{y} \in V(G)^k} \boldsymbol{K}_G(\boldsymbol{x}, \boldsymbol{y}) \boldsymbol{K}'_G(\boldsymbol{y}, \boldsymbol{z})
		\eqqcolon (\boldsymbol{K}_G \cdot \boldsymbol{K}'_G)(\boldsymbol{x}, \boldsymbol{z}).
		\)
		One may similarly compose a graph in $\mathcal{D}(k,k)$ with a graph in $\mathcal{D}(k)$ obtaining one in $\mathcal{D}(k)$.
		This operation corresponds to the matrix-vector product.
	\end{description}

	\subsection{Labelled graphs of bounded treewidth and pathwidth}
	\label{sec:labelled}
	
	Let $F$ be a graph. A \emph{tree decomposition} of $F$ is a pair $(T, \beta)$ where $T$ is a tree and $\beta \colon V(T) \to 2^{V(F)}$ is a map such that
	\begin{enumerate}
		\item the union of the $\beta(t)$ for $t \in V(T)$ is equal to $V(F)$,
		\item for every edge $uv \in E(F)$, there exists $t \in V(T)$ such that $\{u,v\} \subseteq \beta(t)$,
		\item for every vertex $u \in V(F)$, the set of vertices $t \in V(T)$ such that $u \in \beta(t)$ is connected in~$T$.
	\end{enumerate}
	The \emph{width} of $(T, \beta)$ is the maximum over all $\lvert \beta(t) \rvert - 1$ for $t \in V(T)$.
	The \emph{treewidth} $\tw(F)$ of $F$ is the minimum width of a tree decomposition of $F$.
	A \emph{path decomposition} is a tree decomposition $(T, \beta)$ where $T$ is a path. 
    The \emph{pathwidth} $\pw(F)$ of $F$ is the minimum width of a path decomposition of $F$.

    We will use tree decompositions satisfying the following additional properties.

    \begin{lemma}[{\cite[Lemma~8]{bodlaender_partial_1998}}] \label{lem:treedec1}
		Let $k \geq 1$ and $F$ be a graph such that $|V(F)| \geq k$.
		If $\tw(F) \leq k-1$, then $F$ admits a tree decomposition $(T, \beta)$ such that
		\begin{enumerate}
			\item $|\beta(t)| = k$ for all $t \in V(T)$,\label{tw1}
			\item $|\beta(s) \cap \beta(t)| = k-1$ for all $st \in E(T)$.\label{tw2}
		\end{enumerate}
		If $\pw(F) \leq k-1$, then $F$ admits a path decomposition $(T, \beta)$ satisfying \cref{tw1,tw2}.
	\end{lemma}
    
    We use the following family of distinctly labelled graphs, see also \cite[Definition~7]{seppelt_algorithmic_2024}.
    We will not rely on the concrete technical conditions and use them only via \cref{lem:twk-gen}.
	
	\begin{definition}\label{def:twk}
		Let $k \geq 1$.
		Define $\mathcal{TW}(k)$ as the set of all $\boldsymbol{F} = (F, \boldsymbol{u}) \in \mathcal{D}(k)$ such that
		$F$ admits a tree decomposition $(T, \beta)$ of width $\leq k-1$ 
		satisfying the assertions of \cref{lem:treedec1}.
	\end{definition}
	
	It follows that the underlying unlabelled graphs of the labelled graphs in $\mathcal{TW}(k)$ are exactly the graphs of treewidth $\leq k-1$ on $\geq k$ vertices.
	The family $\mathcal{TW}(k)$ is generated by certain small building blocks under series composition and gluing as follows:
	Let $\boldsymbol{1} \in \mathcal{TW}(k)$ be the distinctly $k$-labelled graph on $k$ vertices without any edges.
	For $i \in [k]$, let $\boldsymbol{J}^i = (J^i, (1, \dots, k), \allowbreak (1, \dots, i-1, \widehat{i}, i+1, \dots, k))$ the distinctly $(k,k)$-bilabelled graph with $V(J^i) \coloneqq [k] \cup \{\widehat{i}\}$ and $E(J^i) \coloneqq \emptyset$.
	Writing $\binom{[k]}{2}$ for the set of pairs of distinct elements in $[k]$,
	let $\boldsymbol{A}^{ij} = (A^{ij}, (1, \dots, k), (1, \dots, k))$ for $ij \in \binom{[k]}{2}$ be the distinctly $(k,k)$-bilabelled graph with $V(A^{ij}) \coloneqq [k]$ and $E(A^{ij}) \coloneqq \{ij\}$.
	These graphs are depicted in \cref{fig:basal}.
	Let $\mathcal{B}(k) \coloneqq \{\boldsymbol{J}^i \mid i \in [k]\} \allowbreak \cup \{\boldsymbol{A}^{ij} \mid ij \in \binom{[k]}{2}\}$.

	\begin{figure}
		\begin{subfigure}{.3\linewidth}
			\centering
			\begin{tikzpicture}
				\node (a) [vertex] {};
				\node (ab) [vertex,below of=a] {};
				\node (b) [below of=ab, yshift=.5cm] {\vdots};
				\node (f) [below of=b,vertex, yshift=.3cm] {};
				\node (fb) [below of=f,vertex] {};
				
				\node (a1) [lbl,left of=a] {$1$};
\draw [color=lipicsLightGray] (a) -- (a1);

				\node (ab1) [lbl,left of=ab] {$2$};
\draw [color=lipicsLightGray] (ab) -- (ab1);
\node (f1) [lbl,left of=f] {$k-1$};
\draw [color=lipicsLightGray] (f) -- (f1);
\node (fb1) [lbl,left of=fb] {$k$};
\draw [color=lipicsLightGray] (fb) -- (fb1);
\end{tikzpicture}
			\caption{$\boldsymbol{1} \in \mathcal{TW}(k)$.}
			\label{fig:one}
		\end{subfigure}
		\begin{subfigure}{.3\linewidth}
			\centering
			\begin{tikzpicture}
				\node (a) [vertex] {};
				\node (b) [below of=a, yshift=.5cm] {\vdots};
				\node (c) [below of=b,vertex, yshift=.3cm] {};
				\node (d) [below of=c,vertex] {};
				\node (e) [below of=d, yshift=.5cm] {\vdots};
				\node (f) [below of=e,vertex, yshift=.3cm] {};
				\draw [ultra thick] (c) edge [bend left] (d); 
				\node (e) [below of=c, yshift=.6cm] {\vdots};
				
				\node (a1) [lbl,left of=a] {$1$};
				\node (a2) [lbl,right of=a] {$1$};
				\draw [color=lipicsLightGray] (a) -- (a1);
				\draw [color=lipicsLightGray] (a) -- (a2);
				
				\node (c1) [lbl,left of=c] {$i$};
				\node (c2) [lbl,right of=c] {$i$};
				\draw [color=lipicsLightGray] (c) -- (c1);
				\draw [color=lipicsLightGray] (c) -- (c2);
				\node (d1) [lbl,left of=d] {$j$};
				\node (d2) [lbl,right of=d] {$j$};
				\draw [color=lipicsLightGray] (d) -- (d1);
				\draw [color=lipicsLightGray] (d) -- (d2);
				\node (f1) [lbl,left of=f] {$k$};
				\node (f2) [lbl,right of=f] {$k$};
				\draw [color=lipicsLightGray] (f) -- (f1);
				\draw [color=lipicsLightGray] (f) -- (f2);
			\end{tikzpicture}
			\caption{$\boldsymbol{A}^{ij} \in \mathcal{D}(k,k)$.}
		\end{subfigure}
		\begin{subfigure}{.3\linewidth}
			\centering
			\begin{tikzpicture}
				\node (a) [vertex] {};
				\node (b) [below of=a, yshift=.5cm] {\vdots};
				\node (c) [below of=b,vertex, yshift=.3cm] {};
				\node (dx) [below of=c,vertex, xshift=-.5cm, yshift=.5cm] {};
				\node (dy) [below of=c,vertex, xshift=.5cm, yshift=.5cm] {};
				\node (dz) [below of=dx, xshift=.5cm,vertex, yshift=.5cm] {};
				\node (e) [below of=dz, yshift=.5cm] {\vdots};
				\node (f) [below of=e,vertex, yshift=.3cm] {};

				\node (a1) [lbl,left of=a,xshift=-.5cm] {$1$};
				\node (a2) [lbl,right of=a,xshift=.5cm] {$1$};
				\draw [color=lipicsLightGray] (a) -- (a1);
				\draw [color=lipicsLightGray] (a) -- (a2);
				
				\node (c1) [lbl,left of=c,xshift=-.5cm] {$i-1$};
				\node (c2) [lbl,right of=c,xshift=.5cm] {$i-1$};
				\draw [color=lipicsLightGray] (c) -- (c1);
				\draw [color=lipicsLightGray] (c) -- (c2);
				\node (d1) [lbl,left of=dx] {$i$};
				\node (d2) [lbl,right of=dy] {$i$};
				\draw [color=lipicsLightGray] (d1) to (dx);
				\draw [color=lipicsLightGray] (d2) to (dy);
				
				\node (d3) [lbl,left of=dz,xshift=-.5cm] {$i+1$};
				\node (d4) [lbl,right of=dz,xshift=.5cm] {$i+1$};
				\draw [color=lipicsLightGray] (dz) -- (d3);
				\draw [color=lipicsLightGray] (dz) -- (d4);
				
				\node (f1) [lbl,left of=f,xshift=-.5cm] {$k$};
				\node (f2) [lbl,right of=f,xshift=.5cm] {$k$};
				\draw [color=lipicsLightGray] (f) -- (f1);
				\draw [color=lipicsLightGray] (f) -- (f2);
			\end{tikzpicture}
			\caption{$\boldsymbol{J}^{i} \in \mathcal{D}(k,k)$.}
		\end{subfigure}
		\caption{The (bi)labelled generators of $\mathcal{TW}(k)$ in wire notation of \cite{mancinska_quantum_2020}. A vertex carries in-label (out-label) $i$ if it is connected to the index~$i$ on the left (right) by a wire. Actual edges and vertices of the graph are depicted in black.}
		\label{fig:basal}
	\end{figure}
	
	\begin{lemma}[{\cite[Lemma 9]{seppelt_algorithmic_2024}}] \label{lem:twk-gen}
		Let $k \geq 1$.
		For every $\boldsymbol{F} \in \mathcal{TW}(k)$, one of the following holds:
		\begin{enumerate}
			\item $\boldsymbol{F} = \boldsymbol{1}$,
			\item $\boldsymbol{F} = \prod_{ij \in A} \boldsymbol{A}^{ij} \cdot \boldsymbol{F}'$ for some $A \subseteq \binom{[k]}{2}$ and $\boldsymbol{F}' \in \mathcal{TW}(k)$ with less edges than $\boldsymbol{F}$,
			\item $\boldsymbol{F} = \boldsymbol{J}^i \cdot \boldsymbol{F}'$ for some $i \in [k]$ and $\boldsymbol{F}' \in \mathcal{TW}(k)$ with less vertices than $\boldsymbol{F}$,
			\item $\boldsymbol{F} = \boldsymbol{F}_1 \odot \boldsymbol{F}_2$ for $\boldsymbol{F}_1, \boldsymbol{F}_2 \in \mathcal{TW}(k)$ on less vertices than~$\boldsymbol{F}$.
		\end{enumerate}
	\end{lemma}
    \begin{proof}
        For the sake of completeness, we give the proof from \cite[Lemma 9]{seppelt_algorithmic_2024}.
		Write $\boldsymbol{F} = (F, \boldsymbol{u})$ and
		let $(T, \beta)$ and $r \in V(T)$ be for $\boldsymbol{F} \in \mathcal{TW}(k)$ as in \cref{def:twk}.
		
		For the second case, let $A = \{ij \in \binom{[k]}{2} \mid u_iu_j \in E(F)\}$.
		Write $F'$ for the graph obtained from $F$ by deleting the edges $u_iu_j$ for $ij \in A$.
		Then $\boldsymbol{F}' \coloneqq (F', \boldsymbol{u}) \in \mathcal{TW}(k)$
		and $\boldsymbol{F} = \prod_{ij \in A} \boldsymbol{A}^{ij} \cdot \boldsymbol{F}'$.
		
		If $A = \emptyset$, distinguish two cases:
		The vertex $r$ has a unique neighbour $s \in V(T)$.
		By the assertions of \cref{lem:treedec1}, 
		there exists a unique $i \in [k]$ such that $u_i \not\in \beta(s)$.
		Let $F'$ be the subgraph of $F$ induced by $\bigcup_{t \in V(T) \setminus \{r\}} \beta(t)$
		and write $\boldsymbol{v} \in V(F')^k$ for the tuple of distinct vertices which coincides with $\boldsymbol{u}$
		everywhere except at the $i$-th entry where it is equal to the unique vertex in $\beta(s) \setminus \beta(r)$.
		Then $\boldsymbol{F}' \coloneqq (F', \boldsymbol{v}) \in \mathcal{TW}(k)$ and $\boldsymbol{F} = \boldsymbol{J}^i \cdot \boldsymbol{F}'$.
		
		Finally, suppose that $r$ has multiple neighbours $s_1, \dots, s_m \in V(T)$.
		Let $T_1$ denote the connected component of $T - \{s_2, \dots, s_m\}$ containing $s_1$.
		Let $T_2$ denote the connected component of $T - \{s_1\}$ containing $s_2, \dots, s_m$.
		For $j \in \{1,2\}$, let $F_j$ be the subgraph induced by $\bigcup_{t \in V(T_j)} \beta(t)$. 
		Then $F_1$ contains less vertices than $F$ since the vertex in $\beta(s_2) \setminus \beta(r)$ is not in~$F_1$.
		Symmetrically, $F_2$ contains less vertices than $F$.
		Then $\boldsymbol{F}_j \coloneqq (F_j, \boldsymbol{u}) \in \mathcal{TW}(k)$ for $j \in \{1,2\}$.
		As desired, $\boldsymbol{F} = \boldsymbol{F}_1 \odot \boldsymbol{F}_2$.
	\end{proof}

 	For bounded-pathwidth graph classes, we perform an analogous construction:
    \begin{definition} \label{def:pwk}
		For $k\geq 1$,
		write $\mathcal{PW}(k) \subseteq \mathcal{D}(k)$ for the class of all distinctly $k$-labelled graphs $\boldsymbol{F} = (F, \boldsymbol{u})$ such that there exists a path decomposition $(P, \beta)$ of $F$ where
		\begin{enumerate}
			\item $|\beta(t)| = k$ for all $t \in V(P)$,
			\item $|\beta(s) \cap \beta(t)| = k-1$ for all $st \in E(P)$,
			\item there exists a vertex $r \in V(P)$ of degree at most $1$ such that $\beta(r) = \{u_1, \dots, u_k\}$.
		\end{enumerate}
	\end{definition}
 
	It follows that the underlying unlabelled graphs of the labelled graphs in $\mathcal{PW}(k)$ are exactly the graphs of pathwidth $\leq k-1$ on $\geq k$ vertices.
    The following \cref{lem:pwk-gen} follows from \cref{lem:twk-gen}
    observing that the $\odot$-case does not arise.
	
	\begin{corollary} \label{lem:pwk-gen}
		Let $k \geq 1$.
		For every $\boldsymbol{F} \in \mathcal{PW}(k)$, one of the following holds:
		\begin{enumerate}
			\item $\boldsymbol{F} = \boldsymbol{1}$,
			\item $\boldsymbol{F} = \prod_{ij \in A} \boldsymbol{A}^{ij} \cdot \boldsymbol{F}'$ for some $A \subseteq \binom{[k]}{2}$ and $\boldsymbol{F}' \in \mathcal{PW}(k)$ with less edges than $\boldsymbol{F}$,
			\item $\boldsymbol{F} = \boldsymbol{J}^i \cdot \boldsymbol{F}'$ for some $i \in [k]$ and $\boldsymbol{F}' \in \mathcal{PW}(k)$ with less vertices than $\boldsymbol{F}$.
		\end{enumerate}
	\end{corollary}

	\section{From homomorphism indistinguishability to multiplicity automata equivalence}

    In this section, we reduce homomorphism indistinguishability over recognisable graph classes of bounded treewidth to multiplicity automata equivalence.

    \begin{theorem}[label=thm:main1]
        For $k \in \mathbb{N}$,
        let $\mathcal{F}$ be a $k$-recognisable graph class.
        \begin{enumerate}
            \item If $\mathcal{F}$ has treewidth $< k$,
             $\HomInd(\mathcal{F})$ logspace many-one reduces to \MTA equivalence.
            \item If $\mathcal{F}$ has pathwidth $< k$,
             $\HomInd(\mathcal{F})$ logspace many-one reduces to \MWA equivalence.
        \end{enumerate} 
    \end{theorem}

    \subsection{Recognisability}

    In order to prove the theorem, we give a formal definition of recognisability, see also \cite{courcelle_monadic_1990}.
	
	\begin{definition}[\cite{bojanczyk_definability_2016}] \label{def:recog}
		Let $k \geq 1$.
		For a class of unlabelled graphs $\mathcal{F}$, define the equivalence relation~$\sim_{\mathcal{F}}^k$ on the class of distinctly $k$-labelled graphs~$\mathcal{D}(k)$ by
		letting $\boldsymbol{F}_1 \sim_{\mathcal{F}}^k \boldsymbol{F}_2$ if for all $\boldsymbol{K} \in \mathcal{D}(k)$, it holds that
		\[
		\soe (\boldsymbol{K} \odot \boldsymbol{F}_1) \in \mathcal{F} \iff \soe (\boldsymbol{K} \odot \boldsymbol{F}_2) \in \mathcal{F}.
		\]
		The class $\mathcal{F}$ is \emph{$k$-recognisable} if $\sim_{\mathcal{F}}^k$
		has finitely many equivalence classes.
		The number of classes of~$\sim_{\mathcal{F}}^k$ is the \emph{$k$-recognisability index} of $\mathcal{F}$.
	\end{definition}

    To parse \cref{def:recog}, first recall that $\boldsymbol{K} \odot \boldsymbol{F}_1$ is the $k$-labelled graph obtained by gluing $\boldsymbol{K}$ and $\boldsymbol{F}_1$ together at their labelled vertices. 
	The $\soe$-operator drops the labels yielding unlabelled graphs.
	Intuitively, $\boldsymbol{F}_1 \sim^k_{\mathcal{F}} \boldsymbol{F}_2$ 
	iff both or neither of their underlying unlabelled graphs are in $\mathcal{F}$ 
	and the positions of the labels in $\boldsymbol{F}_1$ and $\boldsymbol{F}_2$ 
	is equivalent with respect to membership in $\mathcal{F}$.
    See \cite{seppelt_algorithmic_2024} for examples.

    Courcelle~\cite{courcelle_monadic_1990} proved that every \CMSO-definable graph class is \emph{recognisable}, 
    i.e.\ it is $k$-recognisable for every $k \in \mathbb{N}$.
    Conversely, Boja\'nczyk and Pilipczuk \cite{bojanczyk_definability_2016} proved that, if a recognisable class $\mathcal{F}$ has bounded treewidth, then it is \CMSO-definable.
    Furthermore, it holds that every $k$-recognisable graph class of treewidth $\leq k$ is \CMSO-definable \cite[Section~6]{bojanczyk_optimizing_2017}.
    We use recognisability via the following lemma.

     \begin{lemma} \label{lem:fefvau}
		For $\boldsymbol{F}, \boldsymbol{F}', \boldsymbol{F}_1, \boldsymbol{F}_2, \boldsymbol{F}'_1, \boldsymbol{F}'_2 \in \mathcal{D}(k)$, $\boldsymbol{L} \in \mathcal{D}(k,k)$,
		\begin{enumerate}
			\item if $\boldsymbol{F}_1 \sim^k_{\mathcal{F}} \boldsymbol{F}'_1$ and $\boldsymbol{F}_2 \sim^k_{\mathcal{F}} \boldsymbol{F}'_2$ 
			then $\boldsymbol{F}_1 \odot \boldsymbol{F}_2 \sim^k_{\mathcal{F}} \boldsymbol{F}'_1 \odot \boldsymbol{F}'_2$,
			\item if $\boldsymbol{F} \sim^k_{\mathcal{F}} \boldsymbol{F}'$ then $\boldsymbol{L} \cdot \boldsymbol{F} \sim^k_{\mathcal{F}} \boldsymbol{L} \cdot \boldsymbol{F}'$.
		\end{enumerate}
	\end{lemma}
    \begin{proof}
		Let $\boldsymbol{K} = (K, \boldsymbol{u}) \in \mathcal{D}(k)$ be arbitrary. 
		Then $\soe((\boldsymbol{K} \odot \boldsymbol{F}_1) \odot \boldsymbol{F}_2) \in \mathcal{F} \Leftrightarrow \soe((\boldsymbol{K} \odot \boldsymbol{F}_1) \odot \boldsymbol{F}'_2) \in \mathcal{F}
		\Leftrightarrow \soe(\boldsymbol{K} \odot \boldsymbol{F}'_1 \odot \boldsymbol{F}'_2) \in \mathcal{F}$.
		
		For a $(k,k)$-bilabelled graph $\boldsymbol{L} = (L, \boldsymbol{u}, \boldsymbol{v})$, write $\boldsymbol{L}^* \coloneqq (L, \boldsymbol{v}, \boldsymbol{u})$ for the $(k,k)$-bilabelled graph obtained by swapping the in- and out-labels.
		Then $\soe(\boldsymbol{K} \odot (\boldsymbol{L} \cdot \boldsymbol{F})) = \soe( (\boldsymbol{L}^* \cdot \boldsymbol{K}) \odot \boldsymbol{F})$.
		Thus the second claim follows from the first.
	\end{proof}

    Before proceeding to the proof of \cref{thm:main1}, we give the following examples for recognisability.

    \begin{example} \label{ex:paths}
		The class $\mathcal{W}$ of all paths is $1$-recognisable.
		Its $1$-recognisability index is $4$.
		The equivalence classes are described by the representatives in \cref{fig:paths}.
	\end{example}
	
	\begin{figure}
		\begin{subfigure}[b]{.2\linewidth}
			\centering
			\begin{tikzpicture}[node distance=.7cm]
				\node (a) [vertex] {};
				\node (a1) [lbl,left of=a] {$1$};
				\draw [color=lightgray] (a) -- (a1);
			\end{tikzpicture}
			\caption{$\boldsymbol{1}$.}
		\end{subfigure}
		\begin{subfigure}[b]{.25\linewidth}
			\centering
			\begin{tikzpicture}[node distance=.7cm]
				\node (a) [vertex] {};
				\node (b) [vertex, right of=a] {};
				\draw [ultra thick] (a) -- (b); 
				\node (a1) [lbl,left of=a] {$1$};
				\draw [color=lightgray] (a) -- (a1);
			\end{tikzpicture}
			\caption{$\boldsymbol{P}$.}
		\end{subfigure}	
		\begin{subfigure}[b]{.25\linewidth}
			\centering
			\begin{tikzpicture}[node distance=.7cm]
				\node (a) [vertex] {};
				\node (b) [vertex, right of=a] {};
				\node (c) [vertex, above of=a] {};
				\draw [ultra thick] (a) -- (b); 
				\draw [ultra thick] (a) -- (c); 
				\node (a1) [lbl,left of=a] {$1$};
				\draw [color=lightgray] (a) -- (a1);
				
			\end{tikzpicture}
			\caption{$\boldsymbol{Q}$.}
		\end{subfigure}	
		\begin{subfigure}[b]{.25\linewidth}
			\centering
			\begin{tikzpicture}[node distance=.7cm]
				\node (a) [vertex] {};
				\node (b) [vertex, right of=a] {};
				\node (c) [vertex, above of=a] {};
				\draw [ultra thick] (a) -- (b); 
				\draw [ultra thick] (a) -- (c); 
				\draw [ultra thick] (b) -- (c);
				\node (a1) [lbl,left of=a] {$1$};
				\draw [color=lightgray] (a) -- (a1);
				
			\end{tikzpicture}
			\caption{$\boldsymbol{C}$.}
		\end{subfigure}		
		\caption{Representatives for $\sim^1_{\mathcal{W}}$ and $\mathcal{W}$ the class of paths from \cref{ex:paths}.}
		\label{fig:paths}
	\end{figure}
	
	\begin{proof}
		To show that the labelled graphs in \cref{fig:paths} cover all equivalence classes, let $\boldsymbol{F} = (F, u)$ be arbitrary.
		If $F$ is not a path then $\boldsymbol{F} \sim^1_{\mathcal{W}} \boldsymbol{C}$. 
		Indeed, for every $\boldsymbol{K} \in \mathcal{D}(1)$, $F$ is a subgraph of $\soe(\boldsymbol{K} \odot \boldsymbol{F})$. Hence, regardless of $\boldsymbol{K}$, both $\soe(\boldsymbol{K} \odot \boldsymbol{F})$ and $\soe(\boldsymbol{K} \odot \boldsymbol{C})$ are not paths.
		If $F$ is a path, then $\boldsymbol{F}$ and $\boldsymbol{1}$, $\boldsymbol{P}$, or $\boldsymbol{Q}$ are equivalent depending on whether the degree of $u$ is $0$, $1$, or $2$.
		
		To show that the representatives in \cref{fig:paths} are in distinct classes, observe for example that $\soe(\boldsymbol{P} \odot \boldsymbol{P}) \in \mathcal{W}$ while $\soe(\boldsymbol{P} \odot \boldsymbol{Q}) \not\in \mathcal{W}$, thus $\boldsymbol{P} \not\sim^1_{\mathcal{W}} \boldsymbol{Q}$.
		Similarly, $\soe(\boldsymbol{1} \odot \boldsymbol{Q}) \in \mathcal{W}$ whereas $\soe(\boldsymbol{P} \odot \boldsymbol{Q}) \not\in \mathcal{W}$, thus $\boldsymbol{1} \not\sim^1_{\mathcal{W}} \boldsymbol{P}$.
	\end{proof}
	A more involved example is the following. Analogously, one may argue that every class defined by forbidden minors is recognisable.
	\begin{example}
		Let $\mathcal{F}$ be the family of $H$-subgraph-free graphs for some graph $H$.
		Then $\mathcal{F}$ is $k$-recognisable for every $k \geq 1$.
	\end{example}
	\begin{proof}
		Suppose $H$ has $m$ vertices.
		For a distinctly $k$-labelled graph $\boldsymbol{F} = (F, \boldsymbol{u})$, consider the set $\mathcal{H}(\boldsymbol{F})$ of (isomorphism types of) distinctly $k$-labelled graphs $\boldsymbol{F}' = (F', \boldsymbol{u})$ where $F'$ is a subgraph of $F$ such that $V(F') \supseteq \{u_1, \dots, u_k\}$ has at most $k+m$ vertices.
		Clearly, there are only finitely many possible sets $\mathcal{H}(\boldsymbol{F})$.
		Furthermore, if $\mathcal{H}(\boldsymbol{F}_1) = \mathcal{H}(\boldsymbol{F}_2)$ then $\boldsymbol{F}_1 \sim^k_{\mathcal{F}} \boldsymbol{F}_2$.
		Indeed, if $\boldsymbol{K}  \in \mathcal{D}(k)$ is such that $\soe(\boldsymbol{K} \odot \boldsymbol{F}_1)$ contains $H$ as a subgraph then so does $\soe(\boldsymbol{K} \odot \boldsymbol{F}_2)$ since $\boldsymbol{F}_1$ and $\boldsymbol{F}_2$ contain the same subgraphs on $k+m$ vertices containing their labelled vertices.
	\end{proof}

    \subsection{Proof of \cref{thm:main1}}

    The proof of \cref{thm:main1}, formally conducted below,
    works by constructing, given a fixed graph class $\mathcal{F}$ and input graphs $G$ and $H$ to $\HomInd(\mathcal{F})$,
    three \MWAs in the bounded-pathwidth case and three \MTAs in the bounded-treewidth case called $\mathcal{A}_{\mathcal{F}}$, $\mathcal{A}_G$, $\mathcal{A}_H$.
    All three automata read words/trees over the alphabet comprising distinctly $(k,k)$-bilabelled graphs representing a single bag of a path/tree decomposition.
    
    The first automaton $\mathcal{A}_{\mathcal{F}}$
    depends only on the graph class $\mathcal{F}$ 
    and does not make use of multiplicities, i.e.\ it is a deterministic word/tree automaton.
    Its role is to recognise the graph class $\mathcal{F}$ among all graphs of pathwidth/treewidth less than $k$.
    The construction of this automaton dates back to Courcelle \cite{courcelle_monadic_1990}.
    Its states are the equivalence classes of $\sim^k_{\mathcal{F}}$.
    \Cref{lem:fefvau} ensures that series and parallel composition, 
    i.e.\ the operations used to compose labelled graphs, 
    respect these equivalence classes.

    The automata $\mathcal{A}_G$ and $\mathcal{A}_H$ depend only on $G$ and $H$, respectively, and are both constructed in the same way. 
    Their role is to compute the homomorphism tensors of labelled graphs of bounded pathwidth/treewidth and thus make full use of multiplicities.
    To that end, they assign to a letter of the input alphabet, i.e.\ a bilabelled graph $\boldsymbol{L} \in \mathcal{D}(k,k)$, the homomorphism matrix $\boldsymbol{L}_G$, respectively $\boldsymbol{L}_H$, as its weight matrix.
    The correspondence between operations on bilabelled graphs and homomorphism tensors, see \cref{operations},
    ensures that numbers computed by the automata are homomorphism counts from graphs of bounded pathwidth/treewidth into $G$ and $H$.
    Note that $\mathcal{A}_G$ and $\mathcal{A}_H$  are equivalent if, and only if, $G$ and $H$ are homomorphism indistinguishable over \emph{all} graphs of pathwidth/treewidth less than $k$.

    The reduction is completed by testing the equivalence of the product automata 
    $\mathcal{A}_{\mathcal{F}} \otimes \mathcal{A}_G$
    and $\mathcal{A}_{\mathcal{F}} \otimes \mathcal{A}_H$.
    For a labelled graph $\boldsymbol{F} \in \mathcal{D}(k)$
    of bounded pathwidth/treewidth with underlying unlabelled graph $F$,
    it is 
    \begin{align*}
        \sem{ \mathcal{A}_{\mathcal{F}} \otimes \mathcal{A}_G} (\boldsymbol{F}) - \sem{ \mathcal{A}_{\mathcal{F}} \otimes \mathcal{A}_H} (\boldsymbol{F}) 
        &= \sem{ \mathcal{A}_{\mathcal{F}}} (\boldsymbol{F}) \cdot \left( \sem{\mathcal{A}_G} (\boldsymbol{F}) - \sem{\mathcal{A}_H} (\boldsymbol{F}) \right) \\
        &= \begin{cases}
            \hom(F, G) - \hom(F, H), & \text{if } F \in \mathcal{F}, \\
            0, & \text{otherwise}.
        \end{cases}
    \end{align*}
    Hence, $G$ and $H$ are homomorphism indistinguishable over $\mathcal{F}$
    if, and only if, the automata $\mathcal{A}_{\mathcal{F}} \otimes \mathcal{A}_G$
    and $\mathcal{A}_{\mathcal{F}} \otimes \mathcal{A}_H$ are equivalent.
    This concludes the proof sketch of \cref{thm:main1}.
    
    \begin{proof}[Proof of \cref{thm:main1}]
        First consider the bounded pathwidth case.
        We define three \MWAs: one depending only on the graph class $\mathcal{F}$ and two each depending on one of the input graphs $G$ and $H$, which we are to test for homomorphism indistinguishability over $\mathcal{F}$.
        The final automata will be the product automata of these.
        
        Let $k \in \mathbb{N}$.
        Let $\Sigma \coloneqq \mathcal{B}(k)$ as defined above \cref{lem:twk-gen}.
        Given a $k$-recognisable graph class $\mathcal{F}$
        of pathwidth less than $k$,
        write $S$ for the equivalence classes of $\sim^k_{\mathcal{F}}$.
        By \cref{lem:fefvau},
        for every $\boldsymbol{L} \in \mathcal{D}(k,k)$,
        there exists a map $f_{\boldsymbol{L}} \colon S \to S$
        such that, if $\boldsymbol{F}$ is in the $\sim^k_{\mathcal{F}}$-equivalence class corresponding to $s \in S$, then $ \boldsymbol{F} \cdot \boldsymbol{L}$ is in the equivalence class corresponding to $f_{\boldsymbol{L}}(s)$.
        Write $M_{\boldsymbol{L}} \in \mathbb{Q}^{S \times S}$ for the matrix $\sum_{s \in S} e_s e_{f_{\boldsymbol{L}}(s)}^\top $.
        That is, $M_{\boldsymbol{L}}$ encodes how $\boldsymbol{L}$ acts on the equivalence classes of $\sim^k_{\mathcal{F}}$.
        The map $M \colon \Sigma \to \mathbb{Q}^{S \times S}$ is now given via $\boldsymbol{L} \mapsto M_{\boldsymbol{L}}$.
        Finally, let $\alpha$ be the indicator vector on the equivalence class containing the labelled graph $\boldsymbol{1} \in \mathcal{D}(k)$ and $\eta$ be the indicator vector on the $s \in S$ corresponding to equivalence classes containing graphs from $\mathcal{F}$.
        In the following claim, we slightly abuse notation by identifying $\Sigma^*$ with a subset of $\mathcal{D}(k,k)$ as justified by \cref{lem:pwk-gen}.

        \begin{claim}
            The \MWA $\mathcal{A}_{\mathcal{F}} = (S, \Sigma, M, \alpha, \eta)$
            maps $\boldsymbol{F} \in \Sigma^*$ to a non-zero number if, and only if, $\soe(\boldsymbol{F}) \in \mathcal{F}$.
        \end{claim}
        \begin{claimproof}
        		Let $\boldsymbol{F} = \boldsymbol{D}^1 \cdots \boldsymbol{D}^\ell$ for some $\boldsymbol{D}^1, \dots, \boldsymbol{D}^\ell \in \mathcal{B}(k)$.
        		The value $\sem{\mathcal{A}_{\mathcal{F}}}(\boldsymbol{F})$ is one if the sequence of states $q_0$, $q_{i+1} \coloneqq f_{\boldsymbol{D}^{i+1}}(q_i)$ for $0 \leq i \leq \ell -1$ terminates in an accepting state,
        		and zero otherwise.
        \end{claimproof}

        The second \MWA $\mathcal{A}_G$ is given by states $S_G \coloneqq V(G)^k$,
        transition matrices are given by $M_G \colon \boldsymbol{L} \mapsto \boldsymbol{L}_G$, and vectors $\alpha'^\top \coloneqq \eta' \coloneqq \boldsymbol{1}_G \in \mathbb{Q}^{V(G)^k}$.
        Define $\mathcal{A}_H$ analogously.

        \begin{claim}
            For every $\boldsymbol{F} \in \Sigma^*$,
            $\sem{\mathcal{A}_G}(\boldsymbol{F}) = \hom(\soe(\boldsymbol{F}), G)$.
        \end{claim}
        \begin{claimproof}
        	Let $\boldsymbol{F} = \boldsymbol{D}^1 \cdots \boldsymbol{D}^\ell$ for some $\boldsymbol{D}^1, \dots, \boldsymbol{D}^\ell \in \mathcal{B}(k)$.
        	The value $\sem{\mathcal{A}_G}(\boldsymbol{F})$ is equal to $\boldsymbol{1}_G^\top  \boldsymbol{D}^1_G \cdots \boldsymbol{D}^\ell_G \boldsymbol{1}_G = \hom(\soe(\boldsymbol{F}), G)$ by the correspondence of series composition to matrix products and sum-of-entries to unlabelling.
        \end{claimproof}

        Clearly, the automata $\mathcal{A}_G$ and $\mathcal{A}_H$ can be computed in logspace.
        By the claims, the product automata 
        $\mathcal{A}_{\mathcal{F}} \otimes \mathcal{A}_G $
        and $\mathcal{A}_{\mathcal{F}} \otimes \mathcal{A}_H $
        are equivalent if, and only if, $G$ and $H$ are homomorphism indistinguishable over the graphs $F \in \mathcal{F}$ on at least $k$ vertices.
        The product automata can be computed in logspace \cite[Proposition~2]{marusic_complexity_2015}.
        Whether $G$ and $H$ are homomorphism indistinguishable over the graphs $F \in \mathcal{F}$ on at most $k$ vertices can be done in logspace with brute force.

		In the bounded treewidth case, we proceed analogously.
		We first define an \MTA $\mathcal{A}_{\mathcal{F}}$.
		For $k \in \mathbb{N}$
		and a $k$-recognisable graph class $\mathcal{F}$,
		let $S$ denote the set of equivalence classes of~$\sim^k_{\mathcal{F}}$.
		Let $\Omega$ be given by $\Omega_0\coloneqq\{\mathbf{1}\}$,  $\Omega_1 \coloneqq \mathcal{B}(k)$
		and $\Omega_2 \coloneqq \{{\odot}\}$.
		The tree representation $\mu$ is given by the union of 
        $\Omega_0 \to \mathbb{Q}^{1 \times S}$, $\boldsymbol{1} \mapsto \alpha$ and
        $\Omega_1 \to \mathbb{Q}^{S \times S}$, $\boldsymbol{L} \mapsto M_{\boldsymbol{L}}$ as defined above, 
		and
		$\Omega_2 \to \mathbb{Q}^{S^2 \times S}$,
		${\odot} \mapsto \sum_{s, s' \in S} e_{(s,s')} e_{g(s,s')}^\top$ where $g \colon S\times S \to S$ is the map such that, if $\boldsymbol{F}$ and $\boldsymbol{F}'$ belong to $s,s'$ respectively, then $\boldsymbol{F} \odot \boldsymbol{F}'$ belongs to $g(s,s')$.
		Let $\eta$ be the indicator vector on the $s \in S$ corresponding to equivalence classes containing graphs from $\mathcal{F}$.
        \begin{claim}
            The \MTA $\mathcal{A}_{\mathcal{F}}$
            maps $\boldsymbol{F} \in T_\Omega$ to a non-zero number if, and only if, $\soe(\boldsymbol{F}) \in \mathcal{F}$.
        \end{claim}
        
        The second \MTA $\mathcal{A}_G$ is given by states $S_G \coloneqq V(G)^k$,
        mapping of the constant $\Omega_0 \to \mathbb{Q}^{1 \times S}$, $\boldsymbol{1} \mapsto \boldsymbol{1}_G^\top \in \mathbb{Q}^{1\times V(G)^k}$,
        and
        arity-$1$ transition matrices  $\Omega_1 \to \mathbb{Q}^{V(G)^k \times V(G)^k}$, $\boldsymbol{L} \mapsto \boldsymbol{L}_G$. 
        The transition matrix of $\odot$ is $M \in \mathbb{Q}^{S_G^2 \times S_G}$ given by 
        $M((\boldsymbol{x}, \boldsymbol{y}), \boldsymbol{z}) = 1$ if $\boldsymbol{x} = \boldsymbol{y} = \boldsymbol{z}$ and zero otherwise for $\boldsymbol{x}, \boldsymbol{y}, \boldsymbol{z} \in V(G)^k$.
        Thereby, $(\boldsymbol{F}_G \otimes \boldsymbol{F}'_G) M = \boldsymbol{F}_G \odot \boldsymbol{F}'_G$.
        The final vector is  $\eta' \coloneqq \boldsymbol{1}_G \in \mathbb{Q}^{V(G)^k}$.
        Define $\mathcal{A}_H$ analogously.
        \begin{claim}
            For every $\boldsymbol{F} \in T_\Omega$,
            $\sem{\mathcal{A}_G}(\boldsymbol{F}) = \hom(\soe(\boldsymbol{F}), G)$.
        \end{claim}

        As above, the product automata $\mathcal{A}_{\mathcal{F}} \otimes \mathcal{A}_G $
        and $\mathcal{A}_{\mathcal{F}} \otimes \mathcal{A}_H $
        are equivalent if, and only if, $G$ and $H$ are homomorphism indistinguishable over the graphs $F \in \mathcal{F}$ on at least $k$ vertices.
    \end{proof}

    As a corollary, we obtain improved bounds on the size of the smallest graph $F$ whose homomorphism counts distinguish two graphs $G$ and $H$ when they are not homomorphism indistinguishable over a recognizable graph class of bounded treewidth. 
    For bounded pathwidth, we recover the bound computed in \cite{seppelt_algorithmic_2024}.

    \begin{corollary}
        Let $\mathcal{F}$ be a $k$-recognisable graph class
        of $k$-recognisability index~$C$.
        Let $G$ and $H$ be graphs on at most $n$ vertices.
        Suppose that $G \not\equiv_{\mathcal{F}} H$.
        \begin{enumerate}
            \item If $\mathcal{F}$ has treewidth $< k$, then there exists a graph $F \in \mathcal{F}$ on at most $4^{Cn^k} k$ vertices such that $\hom(F, G) \neq \hom(F, H)$.
            \item If $\mathcal{F}$ has pathwidth $< k$, then there exists a graph $F \in \mathcal{F}$ on at most $2 Cn^k + k - 1$ vertices such that $\hom(F, G) \neq \hom(F, H)$.
        \end{enumerate}
    \end{corollary}
    \begin{proof}
        By \cite[Theorem~4.2]{seidl_deciding_1990},
        two \MTAs $\mathcal{A}$ and $\mathcal{B}$ on at most $N$ states are equivalent if, and only if, $\sem{\mathcal{A}}(t) = \sem{\mathcal{B}}(t)$ for all trees $t$ of depth at most $2N$. Since the finitary type considered in the proof of \cref{thm:main1} contains symbols of arity $\leq 2$ only, a tree of depth $\leq 2N$ over this type has at most $2^{2N}$ vertices.
        Each such vertex represents a bag of a tree decomposition and thus contributes at most $k$ vertices to the graph $F$.
        Finally, observe that the number of states in the automata from the proof of \cref{thm:main1} is $N \coloneqq C n^k$.

        In the bounded-pathwidth case, observe that only the nullary symbol $\boldsymbol{1}$ contributes $k$ vertices, the unary symbols $\boldsymbol{A}^{ij}$ and $\boldsymbol{J}^i$ contribute at most $1$ vertex.
    \end{proof}
    
	   \section{Complexity of multiplicity word automata equivalence}

    Equipped with the reduction of homomorphism indistinguishability over recognisable bounded-pathwidth graph classes to \MWA equivalence from \cref{thm:main1},
    we now proceed to pinpoint the complexity of the latter problem.
Towards \cref{thm:main3}, we first show containment in \CL, thus improving on \DET as shown in \cite{tzeng_path_1996}, see also \cite{kiefer_notes_2020}.

    \begin{lemma}\label{lem:mwa-in-cl}
        \MWA equivalence is in \CL.    
    \end{lemma}

     \begin{proof}
        The class \CL is closed under logspace many-one reductions as shown in \cite[Theorem~16]{allender_relationships_1996}.
        Thus, it suffices to reduce \MWA equivalence to a problem in \CL.
        We first reduce the equivalence of two \MWAs $\Ac$ and $\Ac'$  over $\Sigma$ to the equivalence of \MWA $\Ac \ominus \Ac'$ and the zero automaton over $\Sigma$. The operation $\ominus$ is clearly computable in logspace \cite[Proposition~2]{marusic_complexity_2015}.
        In the following \cref{clm:linear-characterisation}, the last problem further reduces in logspace to a matrix rank bound verification.
        We use matrices over $\Rat$, however, by  {\cite[Section~2, Remark]{allender_complexity_1999}}, this reduces in logspace to matrices over $\mathbb Z$.
        By {\cite[Proposition 2.5]{allender_complexity_1999}}, the matrix rank bound verification is \CL-complete.
        Thus, it remains to prove the following \cref{clm:linear-characterisation}.
        \begin{claim}\label{clm:linear-characterisation}
            There are logspace-computable functions $N$ and $r$ such that,
            for each \MWA~$\Ac$, it holds that $\Ac$ is equivalent to the zero automaton if, and only if, the matrix $N(\Ac)$ has rank less than $r(\Ac)$.
        \end{claim}
        
        Fix an \MWA $\Ac = (S, \Sigma, M, \alpha, \eta)$ with $n \coloneqq |S|$ states, and denote its square $\Ac \otimes \Ac$ by $\Ac_2 = (S_2, \Sigma, M_2, \alpha_2, \eta_2)$.
        By \cite[Proposition~2.2]{kiefer_notes_2020},
        The \MWA $\Ac$ is equivalent to the zero automaton if, and only if, 
        \begin{equation}
            \sum_{w \in \Sigma^{\le {n-1}}} \sem{\Ac_2}(w) = \sum_{w \in \Sigma^{\le {n-1}}} \left(\sem{\Ac}(w)\right)^2 = 0,\label{eq:sqsum}
        \end{equation}
        Let us denote the sum of all transition matrices by
        \begin{align*}
            T_2 \coloneqq \sum_{a \in \Sigma} M_2(a) = \sum_{a \in \Sigma} M(a)\otimes M(a).
        \end{align*}

        Each entry of $T_2$ is a sum of $|\Sigma|$ products of two entries of the input.
        Multiplication of two numbers, iterated addition of a linear number of numbers,
        and iterated multiplication (the latter needed to bring rational entries over a
        common denominator) are computable in logspace-uniform $\NCone$
        \cite[Theorem~1.3]{chiu_division_2010}, \cite{hesse_uniform_2002}.
        Since logspace-uniform $\NCone \subseteq \LOGSPACE$ \cite{borodin_relating_1977}
        and logspace-computable functions are closed under composition,
        $\Ac_2$ and $T_2$ are computable in logspace.

        We rewrite the expression in \cref{eq:sqsum} as follows
        \begin{align}
            \sum_{w}\sem{\Ac_2}(w)
            &= \sum_{k=0}^{n-1}\;\sum_{a_1\cdots a_k} \alpha_2\, M_2(a_1\cdots a_k)\, \eta_2 \notag\\
            &= \sum_{k=0}^{n-1} \alpha_2\left(\sum_{a_1\cdots a_k} M(a_1\cdots a_k)\otimes M(a_1\cdots a_k)\right)\eta_2 \notag\\
            &= \sum_{k=0}^{n-1} \alpha_2\left(\sum_{a_1\cdots a_k}\left(M(a_1)\otimes M(a_1)\right)\cdots \left(M(a_k)\otimes M(a_k)\right)\right)\eta_2 \notag\\
            &= \sum_{k=0}^{n-1} \alpha_2\, T_2^{\,k}\, \eta_2,
            \label{eq:ksqsum}
        \end{align}
        where $a_1\cdots a_k$ ranges over words in $\Sigma^k$.
        In the third equality, we used the distributivity of the Kronecker product over matrix multiplication.
Next, we define a~matrix $A \in \Rat^{(n^3 + 1) \times n^3}$ and a~vector $b\in \Rat^{n^3 + 1}$ as follows
        \begin{equation*}
            A \coloneqq
            \left(
                \begin{array}{ccccc}
                I_{n^2} & 0 & \cdots & \cdots & 0 \\
                - T_2 & I_{n^2} & \ddots &  & \vdots \\
                0 & -T_2 & \ddots & \ddots & \vdots \\
                \vdots & \ddots & \ddots & I_{n^2} & 0 \\
                0 & \cdots & 0 & -T_2 & I_{n^2}\\
                \alpha_2 & \cdots & \alpha_2 &  \alpha_2 & \alpha_2
                \end{array}
            \right),
            \qquad
            b \coloneqq
            \left(
            \begin{array}{c}
                \eta_2 \vphantom{T_{n^2}}\\
                0 \vphantom{T_{n^2}} \\
                \vdots \vphantom{0\ddots}\\
                0 \vphantom{T_{n^2}}\\
                0 \vphantom{T_{n^2}}\\
                0 \vphantom{\eta_{n^2}}
            \end{array}
            \right).
        \end{equation*}
        We argue that there exists a solution $x\in \Rat^{n^3}$ of $Ax=b$ if, and only if, the term in  \cref{eq:ksqsum} is equal to zero.
        To that end, group the $n^3$ variables into $n$ blocks $y_i$ of $n^2$ variables.
        The first $n$ equations are $y_1 = \eta_2$.
        The subsequent equations yield $y_{i+1} = T_2 y_i$ for all $1 \leq i \leq n-1$.
        The final $n$ equations simplify to $0 = \alpha_2 y_1 + \dots + \alpha_2 y_n = \sum_{k=0}^{n-1} \alpha_2  T_2^k\eta_2$, as desired.

        We now rephrase the existence of the solution $x$ using matrix rank.
        Note that feasibility of linear systems of equations is complete for the potentially larger complexity class $\LOGSPACE^{\CL}$ \cite{allender_complexity_1999}. It is therefore crucial that the rank of $A$ is controlled.
        Since $A$ is lower triangular with identity matrices on the main diagonal, it has rank $n^3$.
        Note that the rank of the augmented matrix $[A| b]$ is $n^3$ if the solution $x$ exists, and $n^3+1$ otherwise.
        Finally, it suffices to set the functions to $N(\Ac) \coloneqq [A|b]$ and the bound on the rank to be verified to $r(\Ac) \coloneqq n^3 + 1$. Both functions can be computed in logspace.
    \end{proof}

    \section{Hardness of homomorphism indistinguishability over recognisable graph classes}
    
    In this section, we show that homomorphism indistinguishability over recognisable graph classes of bounded pathwidth and of bounded treewidth can be as hard as $\CL$ and as \PIT, respectively.
    
    \subsection{Bounded treewidth and  polynomial identity testing}

    Let us formally introduce \PIT. 
    An \emph{arithmetic circuit} \cite{allender_arithmetic_2004} is a directed acyclic graph whose vertices of in-degree zero are labelled by $0$, $1$, or by variables $X_1, \dots, X_\ell$.
	The internal vertices are labelled by $+$, $-$, or $\times$.
	The \emph{Polynomial Identity Testing (\PIT) problem} asks, given a polynomial $f \in \mathbb{Z}[X_1, \dots, X_\ell]$ represented by an arithmetic circuit, whether it is the zero polynomial.
	By the Schwartz--Zippel lemma \cite{schwartz_fast_1980,goos_probabilistic_1979}, \PIT can be solved in randomised polynomial time, i.e.\ lies in $\coRP$. 
	The existence of a deterministic polynomial-time algorithm for \PIT would have far-reaching consequences for circuit complexity \cite{kabanets_derandomizing_2003}, see also \cite{allender_arithmetic_2004}.
    We show that the same holds for $\HomInd(\mathcal{F})$ for some recognisable graph class $\mathcal{F}$ of bounded treewidth.
     
    \thmRecogTree*

    Given \cite{marusic_complexity_2015} and \cref{thm:main1},
    it remains to reduce \PIT to homomorphism indistinguishability.
    To that end, we start with a class of directed vertex-coloured graphs and employ observations made in \textcite{marusic_complexity_2015,allender_complexity_2009}.

    \begin{lemma}\label{lem:pit-to-hom}
        There exists an $\mathsf{MSO}$-definable class of directed vertex-coloured graphs $\mathcal{F}$ of treewidth $\leq 2$ using a fixed finite set of colors such that \PIT logspace many-one reduces to $\HomInd(\mathcal{F})$.
    \end{lemma}
    \begin{proof}
        Following \cite{marusic_complexity_2015,allender_complexity_2009}, we reduce the following variant of \PIT to $\HomInd(\mathcal{F})$.
        A \emph{variable-free arithmetic circuit} is a finite acyclic vertex-labelled directed multigraph whose vertices have in-degree $0$ or $2$.
        The vertices of in-degree $2$ are called \emph{internal gates} and are labelled with $+$, $-$, or $\times$.
        The vertices of in-degree $0$ are labelled with $0$ or $1$.
        The unique vertex of out-degree $0$ is the output gate of the circuit.

        Each such circuit computes an integer with the intuitive semantics.
        By \cite[Proposition~2.2]{allender_complexity_2009} and \cite[Proposition~13]{marusic_complexity_2015},
        \PIT logspace many-one reduces to the following problem:
        Given two variable-free arithmetic circuits $C_1$, $C_2$ satisfying the following conditions,
        decide whether $C_1$ and $C_2$ represent the same integer.
        \begin{enumerate}
            \item the internal gates are labelled with $+$ or $\times$, i.e.\ there are no subtraction gates,
            \item each gate of height $i$ has precisely two children, which are of height $i-1$,
            \item $+$-gates have even height,
            \item $\times$-gates have odd height,
            \item the output gate has even height.
        \end{enumerate}

        Given a circuit $C$ as above,
        define a graph $G(C)$ with vertex colours $\{0, 1, +, \times, S\}$
        by subdividing each outgoing edge of an internal gate and colouring the resulting vertex with $S$.
        In case of $\times$-gates, the two $S$-vertices are connected by an edge.
        See \cref{fig:circuit-to-graph-transform}.

        \begin{figure}
            \centering
            \begin{subfigure}{.25\linewidth}
                \centering
                \begin{tikzpicture}
                    \node [circle, draw] (top) {+};
                    \node [circle, draw, below of=top, xshift=-1cm, yshift=-.5cm] (b1) {\phantom{+}};
                    \node [circle, draw, below of=top, xshift=1cm, yshift=-.5cm] (b2) {\phantom{+}};
                    \draw [->] (top) -- (b1);
                    \draw [->] (top) -- (b2);
                \end{tikzpicture}
                \vspace{.5cm}
                
                \begin{tikzpicture}
                    \node [circle, draw] (top) {+};
                    \node [circle, draw, below of=top, xshift=-1cm, yshift=-.5cm] (b1) {\phantom{+}};
                    \node [circle, draw, below of=top, xshift=1cm, yshift=-.5cm] (b2) {\phantom{+}};
                    \draw (top) edge [->] node [midway, fill, inner sep=2pt, circle] {} (b1);
                    \draw (top) edge [->] node [midway, fill, inner sep=2pt, circle] {} (b2);
                \end{tikzpicture}
            \end{subfigure}
            \begin{subfigure}{.2\linewidth}
                \centering
                \begin{tikzpicture}
                    \node [circle, draw] (top) {+};
                    \node [circle, draw, below of=top, yshift=-.5cm] (b1) {\phantom{+}};
                    \draw (top) [->] edge [bend left] (b1);
                    \draw (top) [->] edge [bend right] (b1);
                \end{tikzpicture}
                \vspace{.5cm}
                
                \begin{tikzpicture}
                    \node [circle, draw] (top) {+};
                    \node [circle, draw, below of=top, yshift=-.5cm] (b1) {\phantom{+}};
                    \draw (top) edge [bend left, ->] node [midway, fill, inner sep=2pt, circle] {} (b1);
                    \draw (top) edge [bend right, ->] node [midway, fill, inner sep=2pt, circle] {} (b1);
                \end{tikzpicture}
            \end{subfigure}
            \begin{subfigure}{.25\linewidth}
                \centering
                \begin{tikzpicture}
                    \node [circle, draw] (top) {$\times$};
                    \node [circle, draw, below of=top, xshift=-1cm, yshift=-.5cm] (b1) {\phantom{+}};
                    \node [circle, draw, below of=top, xshift=1cm, yshift=-.5cm] (b2) {\phantom{+}};
                    \draw [->] (top) -- (b1);
                    \draw [->] (top) -- (b2);
                \end{tikzpicture}
                \vspace{.5cm}
                
                \begin{tikzpicture}
                    \node [circle, draw] (top) {$\times$};
                    \node [circle, draw, below of=top, xshift=-1cm, yshift=-.5cm] (b1) {\phantom{+}};
                    \node [circle, draw, below of=top, xshift=1cm, yshift=-.5cm] (b2) {\phantom{+}};
                    \draw (top) edge [->] node [midway, fill, inner sep=2pt, circle] (m1) {} (b1);
                    \draw (top) edge [->] node [midway, fill, inner sep=2pt, circle] (m2) {} (b2);
                    \draw (m1) edge [bend right, <->] (m2);
                \end{tikzpicture}
            \end{subfigure}
            \begin{subfigure}{.2\linewidth}
             \centering
                \begin{tikzpicture}
                    \node [circle, draw] (top) {$\times$};
                    \node [circle, draw, below of=top, yshift=-.5cm] (b1) {\phantom{+}};
                    \draw (top) edge [bend left, ->] (b1);
                    \draw (top) edge [bend right, ->] (b1);
                \end{tikzpicture}
                \vspace{.5cm}
                
                \begin{tikzpicture}
                    \node [circle, draw] (top) {$\times$};
                    \node [circle, draw, below of=top, yshift=-.5cm] (b1) {\phantom{+}};
                    \draw (top) edge [bend left, -> ] node [midway, fill, inner sep=2pt, circle] (m1) {} (b1);
                    \draw (top) edge [bend right, ->] node [midway, fill, inner sep=2pt, circle] (m2) {} (b1);
                    \draw (m1) edge [<->] (m2);
                \end{tikzpicture}
            \end{subfigure}
            \caption{How to transform a circuit $C$ in the top row into a graph $G(C)$ in the bottom row.}
            \label{fig:circuit-to-graph-transform}
        \end{figure}

        Consider the family of graphs $F_h$ as defined in \cref{fig:graphclass-f}.
        We show the following claim.
        Here, $\val(g)$ for a gate $g$ denotes the integer computed by $g$.

        \begin{figure}
            \centering
            \begin{subfigure}[b]{.1\linewidth}
                \centering
                \begin{tikzpicture}
                    \node [circle, draw] {$1$};
                \end{tikzpicture}
                \caption{$F_0$.}
            \end{subfigure}
            \quad \quad
            \begin{subfigure}[b]{.3\linewidth}
            \centering
            \begin{tikzpicture}
                \node [circle, draw] (top) {$\times$};

                \node[isosceles triangle,
                        draw,   
                        rotate=90,
                        rotate=0,
                        minimum size =1.5cm, below of=top, yshift=2cm, xshift= -2.5cm] (T1) {};
            
                \node at (T1.center) {$F_{2h}$};
                \draw (top) edge[->]  node [midway, fill, inner sep=2pt, circle ] (m1) {} (T1.east);

                \node[isosceles triangle,
                        draw,   
                        rotate=90,
                        rotate=0,
                        minimum size =1.5cm, below of=top, yshift=0cm, xshift= -2.5cm] (T2) {};
            
                \node at (T2.center) {$F_{2h}$};
                \draw (top) edge [->]  node [midway, fill, inner sep=2pt, circle ] (m2) {} (T2.east);
                \draw (m1) edge [bend right, <->] (m2);
            \end{tikzpicture}
             \caption{$F_{2h+1}$ for $h \geq 0$.}
            \end{subfigure}
            \quad \quad
            \begin{subfigure}[b]{.2\linewidth}
            \centering
            \begin{tikzpicture}
                \node [circle, draw] (top) {$+$};

                \node[isosceles triangle,
                        draw,   
                        rotate=90,
                        rotate=0,
                        minimum size =1.5cm, below of=top, yshift=1cm, xshift= -2.5cm] (T1) {};
            
                \node at (T1.center) {$F_{2h-1}$};
                
                \draw (top) edge [->] node [midway, fill, inner sep=2pt, circle ] (m1) {} (T1.east);
            \end{tikzpicture}
            \caption{$F_{2h}$ for $h \geq 1$.}
            \end{subfigure}
            \caption{The graphs  $F_h$ for $h \geq 0$.}
            \label{fig:graphclass-f}
        \end{figure}

        \begin{claim}\label{claim:circuit-hom}
            Let $C$ be a circuit.
            Let $g \in V(C)$ be a gate at height $h \geq 0$.
            The number of homomorphisms $\hom(F_h, G(C); r \mapsto g)$ from $F_h$ to $G(C)$ which map the root $r$ of $F_h$ to $g$ is equal to $\alpha(h) \cdot \val(g) $ for $\alpha(h) \coloneqq 2^{2^{\lceil h/2 \rceil} - 1}$.
        \end{claim}
        \begin{claimproof}
            By induction on $h$.
            If $h = 0$,
            then $g$ is either labelled $0$ or $1$.
            Since we consider vertex-colour preserving homomorphisms, 
            the claim follows.

            For the inductive step, distinguish two cases:
            If $h$ is odd, then $g$ is a multiplication gate.
            Write $g_1, g_2$ for the two children of $g$.
            Note that it may be that $g_1 = g_2$.
            Write $r_1$ and $r_2$ for the roots of the two copies of $F_{h-1}$ in $F_h$.
            First suppose that $ g_1 = g_2 \eqqcolon g'$.
            Then, by the inductive hypothesis,
            \begin{align*}
                \hom(F_h, G(C); r \mapsto g)
                &=
                2 \hom(F_{h-1}, G(C); r_1 \mapsto g')
                \hom(F_{h-1}, G(C); r_2 \mapsto g') \\
                &= 2 \alpha(h-1)^2 \cdot\val(g')^2 \\
                &= \alpha(h) \cdot\val(g).
            \end{align*}
            Otherwise, i.e.\ if $g_1 \neq g_2$, by the inductive hypothesis,
            \begin{align*}
                \hom(F_h, G(C); r \mapsto g)
                &=
                \hom(F_{h-1}, G(C); r_1 \mapsto g_1)
                \hom(F_{h-1}, G(C); r_2 \mapsto g_2) \\
                &\quad + \hom(F_{h-1}, G(C); r_1 \mapsto g_2)
                \hom(F_{h-1}, G(C); r_2 \mapsto g_1) \\
                &= 2 \alpha(h-1)^2 \cdot\val(g_1) \val(g_2) \\
                &= \alpha(h) \cdot\val(g).
            \end{align*}

            It remains to consider the case when $h$ is even, i.e.\ when $g$ is an addition gate.
            Write $r'$ for the root of $F_{h-1}$ in $F_h$.
            Write $g_1, g_2$ for the two children of $g$.
            If $ g_1 = g_2 \eqqcolon g'$,
            \begin{align*}
                \hom(F_h, G(C); r \mapsto g)
                &= 
                2 \hom(F_{h-1}, G(C); r' \mapsto g')\\
                &= 2 \alpha(h-1) \cdot\val(g') \\
                &= \alpha(h-1) \cdot(\val(g') + \val(g'))\\
                &= \alpha(h) \cdot\val(g).
            \end{align*}
            where the factor $2$ is incurred by the two possible images of the subdivision vertex between $r$ and $r'$.
            Also note that $\alpha(h) = \alpha(h-1)$ when $h$ is even.
            Finally, if $g_1 \neq g_2$,
            \begin{align*}
                \hom(F_h, G(C); r \mapsto g)
                &= 
                \hom(F_{h-1}, G(C); r' \mapsto g_1) + \hom(F_{h-1}, G(C); r' \mapsto g_2) \\
                &= \alpha(h-1) \cdot (\val(g_1) + \val(g_2))\\
                &= \alpha(h) \cdot \val(g). \qedhere
            \end{align*}
        \end{claimproof}

        Define $\widehat{G}(C)$ by making the output gate vertex in $G(C)$ adjacent to a fresh vertex of colour $T$.
        Similarly, define $\widehat{F_h}$ by making the output vertex in $F_h$ adjacent to a fresh vertex of colour $T$.
        Let $\mathcal{F}' \coloneqq \{\widehat{F_h} \mid h \geq 0\}$.
        Finally,
        define $\mathcal{F}$ analogously but without imposing the restriction that all leaves of the graphs have the same depth.
        That is,
        the subgraphs $F_{2h}$ and $F_{2h-1}$ in \cref{fig:graphclass-f} are replaced with any $F_{h'}$ for even and odd $h'$, respectively.
        
        \begin{claim}
        	The graph class $\mathcal{F}$ can be defined in $\mathsf{MSO}$.
        \end{claim}
        \begin{claimproof}
        	The graphs in $\mathcal{F}$ are precisely those directed graphs with vertex colours $0$, $1$, $+$, $\times$, $S$, and $T$ that possess the following $\mathsf{MSO}$-definable properties:
        	\begin{enumerate}
        		\item The underlying undirected graphs are connected and removing any vertex whose overall degree is at least $2$ disconnects the directed graph.\label{prop1}
        		\item Every vertex has precisely one of the colours  $\{0,1,+,\times, S, T\}$.
        		\item There is a unique vertex coloured $T$.
        		\item The vertices of out-degree zero are of colours $0$ or $1$.
        		\item The vertices coloured $+$ have precisely one out-neighbour coloured $S$, which has precisely one out-neighbour coloured $\times$, $0$, or $1$.
        		\item The vertices coloured $\times$ have precisely two out-neighbours coloured $S$, which are mutually adjacent and each have precisely one out-neighbour coloured $+$, $0$, or $1$.
        	\end{enumerate}
            Note that all but the first property are actually first-order.
            Only acyclicity and connectedness in \cref{prop1} require second-order resources.
        \end{claimproof}
        
        Finally, we claim that $C_1$ and $C_2$ represent the same integer if, and only, if $\widehat{G}(C_1)$ and $\widehat{G}(C_2)$ are homomorphism indistinguishable over $\mathcal{F}$.
        Since $C_1$ and $C_2$ have the same height $h$ and the unique $T$-coloured vertex must be mapped to the unique $T$-coloured vertex, it follows from \cref{claim:circuit-hom} that, for $i \in \{1,2\}$,
        \[
            \hom(\widehat{F_h}, \widehat{G}(C_i))
            = \alpha(h) \cdot \val(C_i)
        \]
        Furthermore,
        $
            \hom(\widehat{F_{h'}}, \widehat{G}(C_i)) = 0
        $
        for $h' \neq h$.
        In particular, the graphs in $\mathcal{F} \setminus \mathcal{F}'$ do not admit any homomorphisms to $\widehat{G}(C_i)$.
    \end{proof}

	The following \cref{lem:uncolour} allows to reduce a homomorphism indistinguishability problem over a directed vertex-coloured graph class to one over a class of undirected graphs.

    \begin{lemma}\label{lem:uncolour}
        Let $\mathcal{C}$
        be a class of directed vertex-coloured graphs with finitely many colours.
        Then there exists a graph class $\mathcal{F}$
        such that 
        \begin{enumerate}
            \item $\HomInd(\mathcal{C})$ logspace many-one reduces to $\HomInd(\mathcal{F})$,
            \item if $\mathcal{C}$ has bounded treewidth, then so does $\mathcal{F}$, 
            \item if $\mathcal{C}$ has bounded pathwidth, then so does $\mathcal{F}$, and
            \item if $\mathcal{C}$ is $\mathsf{MSO}$-definable, 
            then  $\mathcal{F}$ is $\mathsf{MSO}_1$-definable.
        \end{enumerate}
    \end{lemma}

    The proof of \cref{lem:uncolour} is based on a construction of \textcite{boker_et_al:LIPIcs.MFCS.2019.54,boeker_complexity_2025}
    involving Kneser graphs.
    For integers $r$ and $s$ such that $1 \leq r \leq s/2$,
    the \emph{Kneser graph} $K(r,s)$ is the graph whose vertices are the $r$-subsets of $[s]$
    and whose edges connect two vertices if, and only if,
    the corresponding subsets are disjoint.
    The crucial property of Kneser graphs, which allows them to simulate colours and edge directions, is that there are homomorphically incomparable \cite{hahn_graph_1997}.

    \begin{proof}[Proof of \cref{lem:uncolour}]
        Let $C$ denote the finite set of vertex colours of graphs occurring in $\mathcal{C}$.
        By \cite{hahn_graph_1997,boker_et_al:LIPIcs.MFCS.2019.54,boeker_complexity_2025},
        there exists a collection of homomorphically incomparable connected Kneser graphs $P$, $Q$, $V$, and $K_c$, for $c \in C$.
        Write $\ell$ for the maximum number of vertices among these graphs.
        For each of these Kneser graphs, pick some vertex which will subsequently be called the \emph{tip} of the Kneser graph.
        We consider the gadgets depicted in \cref{fig:direction-gadget,fig:indicator-gadget}.
        The circles denote the Kneser graphs, the incident filled circles their tips.

        \begin{figure}
            \centering
        	\begin{subfigure}{.4\linewidth}
        		\centering
                \resizebox{!}{.9\height}{
        		\begin{tikzpicture}
        			\node [circle, fill, inner sep=2pt] (u) {};
        			\node [anchor=south] at (u.north) {$u$};
        			
        			\node [circle, fill, inner sep=2pt] at (2,0) (p1) {};
        			\node [circle, fill, inner sep=2pt] at (4,0) (p2) {};
        			\node [circle, fill, inner sep=2pt] at (6,0) (v) {};
        			\node [anchor=south] at (v.north) {$v$};
        			
        			\draw (u) edge node [midway,
        			above] {$P_{10\ell}$} (p1);
        			
        			\draw (p1) edge node [midway,
        			above] {$P_{2\ell}$} (p2);
        			
        			\draw (p2) edge node [midway,
        			above] {$P_{10\ell}$} (v);
        			
        			\node [draw, circle, inner sep=6pt] at (2, 2) (g1) {$P$};
        			\node [circle, fill, inner sep=2pt] at (g1.south) (c1) {}; 
        			\node [draw, circle, inner sep=6pt] at (4, 2) (g2) {$Q$};
        			\node [circle, fill, inner sep=2pt] at (g2.south) (c2) {}; 
        			
        			\draw (p1) edge node [midway,
        			left] {$P_{2\ell}$} (c1);
        			\draw (p2) edge node [midway,
        			right] {$P_{2\ell}$} (c2);
        		\end{tikzpicture}}
        		\caption{Direction gadget $L_{u \to v}$.}
        		\label{fig:direction-gadget}
        	\end{subfigure}
            \begin{subfigure}{.55\linewidth}
            	\centering
                \resizebox{!}{.9\height}{
            	\begin{tikzpicture}
            		\node [circle, fill, inner sep=2pt] (u) {};
            		\node [anchor=south] at (u.north) {$u$};
            		
            		\node [draw, circle, inner sep=6pt] at (-2, 2) (g1) {$K_c$};
            		\node [circle, fill, inner sep=2pt] at (g1.south) (c1) {}; 
            		\node [draw, circle, inner sep=6pt] at (2, 2) (g2) {$V$};
            		\node [circle, fill, inner sep=2pt] at (g2.south) (c2) {}; 
            		
            		\draw (u) edge [bend left] node [midway,
            		right, anchor=south west] {$P_{2\ell}$} (c1);
            		\draw (u) edge [bend right]  node [midway,
            		left, anchor=south east] {$P_{2\ell}$} (c2);
            	\end{tikzpicture}}
            	\caption{Indicator gadget for vertex $u$ of colour $c \in C$.}
            	\label{fig:indicator-gadget}
            \end{subfigure}
            \caption{Gadgets constructed in the proof of \cref{lem:uncolour}.}
        \end{figure}

        For a directed vertex-coloured graph $F$ with colours in $C$,
        let $\widetilde{F}$ denote the uncoloured graph obtained by applying these gadgets to $F$.
        Let $\mathcal{F} \coloneqq \{\widetilde{F} \mid F \in \mathcal{C}\}$.
        The following claims imply the theorem.

        \begin{claim}
            The map $G \mapsto \widetilde{G}$ can be computed in logspace.
        \end{claim}
        \begin{claimproof}
        	The replacement of vertex colours and edge directions by hardcoded gadgets can be implemented in logspace. 
        \end{claimproof}

        \begin{claim}
            For every $F \in \mathcal{C}$,
            $\tw(\widetilde{F}) \leq \ell \cdot \tw(F)$
            and $\pw(\widetilde{F}) \leq \ell \cdot \pw(F)$.
        \end{claim}
        \begin{claimproof}
            Observe that the indicator gadget for a vertex $u$ admits a path decomposition of width $\ell$ such that $u$ lies in both end bags.
            Furthermore,
            the direction gadget for $u\to v$ admits a path decomposition of width $\ell$
            such that $u, v$ lie in both end bags.

            Given a tree decomposition of $F$ of width $\tw(F)$,
            a tree decomposition of $\widetilde{F}$ can be constructed by inserting the tree decompositions of the gadgets appropriately.
            The width of this decomposition is $\ell \cdot \tw(F)$.
            For path decompositions, the argument is analogous.
        \end{claimproof}
        
        \begin{claim}
        	$\mathcal{F}$ can be defined in $\mathsf{MSO}_1$.
        \end{claim}
        \begin{claimproof}
        	First, we argue that the class of graphs $\widetilde{G}$ for arbitrary directed $C$-coloured graphs $G$ can be defined in first-order logic. 
            To that end, note that there is, for every type of gadget, a first-order formula $\phi(x)$ which is true precisely for the vertices that are part of a gadget of this type. 
        	This is because the size of each gadget is constant.
        	Combining these formulas allows to define the graphs of the form $\widetilde{G}$.
        	
        	Secondly, we may define in a graph $\widetilde{G}$ the original vertex colours and edge-directions of the original graph $G$ using first-order logic.
        	The defining $\mathsf{MSO}_1$-formula for $\mathcal{F}$ is then obtained by substitution from the $\mathsf{MSO}$-formula defining $\mathcal{C}$.
        \end{claimproof}

        The following claim follows from \cite[Lemma 17]{boker_et_al:LIPIcs.MFCS.2019.54}.
        See also \cite[Lemma~5.1]{boeker_complexity_2025}.
        
        \begin{claim}\label{cl:hom-tilde}
            For every $F \in \mathcal{C}$,
            there exists a constant $\alpha_F \neq 0$
            such that $\hom(F, G) = \alpha_F \hom(\widetilde{F}, \widetilde{G})$
            for every $G$.
        \end{claim}

       Hence, $G$ and $H$ are homomorphism indistinguishable over $\mathcal{C}$ if, and only if, $\widetilde{G}$ and $\widetilde{H}$ are homomorphism indistinguishable over $\mathcal{F}$.
    \end{proof}

    This concludes the preparations for the proof of \cref{thm:recognisable-trees}.

    \begin{proof}[Proof of \cref{thm:recognisable-trees}]
    	The logspace many-one interreducibility of \MTA equivalence and \PIT was established in \cite[Propositions 12 and 13]{marusic_complexity_2015}.
    	By \cref{lem:pit-to-hom,lem:uncolour},
    	\PIT logspace many-one reduces to $\HomInd(\mathcal{F})$ for some $\mathsf{MSO}_1$\nobreakdash-definable graph class of bounded treewidth.
    	By \cref{thm:main1}, $\HomInd(\mathcal{F})$ logspace many-one reduces to \MTA equivalence.
    \end{proof}

    \subsection{Bounded pathwidth and \texorpdfstring{\CL}{C=L}}

    In this section, we show \cref{thm:pw-hardness} and thereby complete the proof of \cref{thm:main3}.

    \thmRecogPath*

    Our reduction is from the problem \VCP of
    verifying the characteristic polynomial $\chi_A$ of an integer matrix $A$, that is the decision problem
    \begin{align*}
         \left\{(A, c_0, c_1, \dots, c_{n-1}) 
         \ \middle| \ n \in \mathbb{N}, A \in \mathbb{Z}^{n \times n}, c_0, \dots, c_{n-1} \in \mathbb{Z}, \,\,
         \chi_A(\lambda) = \lambda^n + \sum_{i=0}^{n-1} c_i \lambda^i 
         \right\}.
    \end{align*}
    \VCP was shown to be \CL-complete under logspace many-one reductions in \cite[Theorem~3.2]{hoang_complexity_2000}.
    We start with the following \cref{lem:pos-det} by which we treat negative entries.
    
\begin{lemma}\label{lem:pos-det}
    For every pair of matrices $A, B\in \mathbb{Z}^{n \times n}$, there exist logspace-computable matrices 
    $D, E\in \mathbb{N}^{3n \times 3n}$ such that
    \begin{enumerate}
        \item $\chi_A = \chi_B$ if, and only if, $\chi_D = \chi_E$,
        \item $A$ and $B$ are similar if, and only if, $D$ and $E$ are similar.
    \end{enumerate}
\end{lemma}
\begin{proof}
    If $n=0$, then the statement holds trivially. Assume $n>0$
    and denote the matrix of non-negative and negative elements of $A$ by $A^+$ and $A^-$, respectively, so that $A = A^+ - A^-$ and $|A| = A^+ + A^-$.
    Define $B^+$, $B^-$, and $|B|$ analogously.
    We define matrices $D$ and $E$ in $\mathbb{N}^{3n \times 3n}$ via
    \[
        D \coloneqq \begin{pmatrix}
        A^+ & A^- & 0 \\
        A^- & A^+ & 0 \\
        0   &   0 & |B|
        \end{pmatrix},
        \quad
        E \coloneqq \begin{pmatrix}
        B^+ & B^- & 0 \\
        B^- & B^+ & 0 \\
        0   &   0 & |A|
        \end{pmatrix},
        \quad \text{ and } \quad 
        T \coloneqq \begin{pmatrix}
        -I&  I & 0 \\
        I &  I & 0 \\
        0 &  0 & I
        \end{pmatrix}.
    \]
    We consider a similar matrix to $D$
    \[
        TD T^{-1} \coloneqq \begin{pmatrix}
            A &  0 & 0 \\
            0 & |A| & 0 \\
            0 &  0 & |B|
        \end{pmatrix},
        \quad \text{ where } \quad
        T^{-1} \coloneqq \begin{pmatrix}
        -\frac{1}{2}I &  \frac{1}{2}I & 0 \\
        \frac{1}{2}I & \frac{1}{2}I & 0 \\
        0 &  0 & I
        \end{pmatrix}.
    \]
    and analogously the matrix $T E T^{-1}$, which is similar to $E$.
    It follows that $\chi_D = \chi_{T D T^{-1}} = \chi_{A}\cdot \chi_{|A|} \cdot \chi_{|B|}$
    and $\chi_E = \chi_{T E T^{-1}} = \chi_{B}\cdot \chi_{|A|}\cdot  \chi_{|B|}$.
    Since $n>0$, all polynomials in the product are non-zero, and thus cancelling both polynomials by the term $\chi_{|A|} \cdot \chi_{|B|}$, yields the first statement.

    For the second statement, by \cite[Theorem~7.7.3]{roman_advanced_2008} courtesy to \cite{hattingh_show_2018},
    $TD T^{-1}$ and $TE T^{-1}$ are similar if, and only if,
    $A$ and $B$ are similar.
    Clearly, $D$ and $E$ are similar if, and only if, $TD T^{-1}$ and $TE T^{-1}$ are similar.
\end{proof}

\Cref{lem:pos-det} yields that checking whether two matrices with non-negative integral entries have the same characteristic polynomial is \CL-complete.

\begin{theorem}\label{thm:nonsym-char}
    The set $\{(A, B) \mid n\in \mathbb{N}, A, B \in \mathbb{N}^{n \times n}, \chi_A = \chi_B \}$ is $\CL$-complete under logspace many-one reductions.
\end{theorem}
\begin{proof}
    We show $\CL$-hardness and postpone containment in virtue of \cite[Theorem~16]{allender_relationships_1996} to \cref{thm:cycles}.
    For hardness, given an instance $(A,$ $c_0, c_1, \dots, c_{n-1})$ of \VCP, 
    where $A \in \mathbb{Z}^{n \times n}$, $n \in \mathbb{N}$, we need to decide if $q(\lambda) = \lambda^n + \sum_{i=0}^{n-1} c_i \lambda^i$ is equal to the characteristic polynomial $\chi_A$.
    For that, take as $B \in \mathbb{Z}^{n \times n}$ the companion matrix  of $q(x)$, for which it holds that $\chi_B = q$ \cite[Theorem~7.12]{roman_advanced_2008}.
    We use \cref{lem:pos-det} to obtain a pair of non-negative matrices whose characteristic polynomials are equal if, and only if, $\chi_A = \chi_B = q$.
\end{proof}

It remains to show that
non-negative integral entries can be simulated by $\{0,1\}$-entries using suitable gadgets.
By the following \cref{thm:cycles}, the decision problem $\{(A, B) \mid n\in \mathbb{N}, \allowbreak A, B \in \{0,1\}^{n \times n}, \allowbreak  \chi_A = \chi_B \}$ is \CL-complete.
We remark that \cite[Figure~2.1]{toda_counting_1984} describes a similar gadget construction when reducing the problem of computing the determinant of an integer matrix to computing the powers of a $\{-1,0,1\}$-matrix.

\begin{theorem}\label{thm:cycles}
    Homomorphism indistinguishability over directed cycles is $\CL$-complete under logspace many-one reductions.
\end{theorem}
\begin{proof}
    Containment was shown in \cref{obs:paths}.
    For hardness, let $A, B \in \mathbb{N}^{n \times n}$ be an instance of problem that is shown \CL-hard in \cref{thm:nonsym-char}.
    We consider the following two directed weighted graphs $G_w$ and $H_w$ given as adjacency matrices $A$ and $B$, respectively.

           \begin{figure}
        \centering
        \begin{tikzpicture}[yscale=0.75]
            \node [circle, draw, inner sep=4pt] (ua) at (0,4) {};
            \node [circle, draw, inner sep=4pt] (va) at (10,4) {};
            \node [circle, draw, inner sep=4pt] (ub) at (0,2) {};
            \node [circle, draw, inner sep=4pt] (vb) at (10,2) {};
            \node [circle, draw, inner sep=4pt] (uc) at (0,0) {};
            \node [circle, draw, inner sep=4pt] (vc) at (10,0) {};

\foreach \n\px/\py in {
                b/1/2,
                c/1/0, c2/3/0, c3/7/0
            } {
                \node [circle, fill, inner sep=2pt] (d1\n) at (\px, \py) {};
                \node [circle, fill, inner sep=2pt] (d2\n) at (\px+1, \py+.5) {};
                \node [circle, fill, inner sep=2pt] (d3\n) at (\px+1, \py-.5) {};
                \node [circle, fill, inner sep=2pt] (d4\n) at (\px+2, \py) {};
                \draw[->] (d1\n) to (d2\n);
                \draw[->] (d1\n) to (d3\n);
                \draw[->] (d2\n) to (d4\n);
                \draw[->] (d3\n) to (d4\n);
            }

\foreach \n\px/\py in {
                a/1/4, a2/3/4, a3/7/4,
                b2/3/2, b3/7/2
            } {
                \node [circle, fill, inner sep=2pt] (d1\n) at (\px, \py) {};
                \node [circle, fill, inner sep=2pt] (d2\n) at (\px+1, \py+.35) {};
                \node [circle, fill, inner sep=2pt] (d4\n) at (\px+2, \py) {};
                \draw[->] (d1\n) to (d2\n);
                \draw[->] (d2\n) to (d4\n);
            }

            \foreach \n in {a,b,c} {
\node [anchor=south] at (u\n.north) {$u$};
\node [anchor=south] at (v\n.north) {$v$};
                \draw[draw=none] (d4\n2) -- node[midway] {$\dots$} (d1\n3);
                \draw[->] (u\n) to (d1\n);
                \draw[->] (d4\n3) to (v\n);
            }
        \end{tikzpicture}
        \caption{Gadgets $B_{1,b}$, $B_{2,b}$, and $B_{b,b}$ used for replacing the directed edge $uv$.}
        \label{fig:mgadget}
    \end{figure}

    Choose a bit-length $b$ such that $2^{b+1} - 1$ bounds entries of $A$ and $B$.
    We construct a directed simple graph $G$, by starting with vertices of $G_w$.
    Next, for every edge $(u,v)$ of $G_w$ with weight $m$,
    we connect $u$ to $v$ in $G$ with the gadgets given in \cref{fig:mgadget} as follows:
    for the $i$-th non-zero bit of $m$ we add gadget $B_{i,b}$ to $G$ between $u$ and $v$.
    We construct a directed graph $H$ from $H_w$ analogously.
    
    \begin{claim}
    For all $k$, $\hom(\vec{C}_{kb}, G) = \tr(A^k)$, and $\hom(\vec{C}_\ell, G) = 0$ if $b$ does not divide $\ell$.
    \end{claim}
    \begin{claimproof}
    Consider a walk between endpoints $u$ and $v$ of a gadget $B_{i,b}$.
    This walk is of length exactly $b$, because of the orientation.
    Each $B_{i,b}$ contributes $2^i$ distinct walks, so that the collection of gadgets between $u$ and $v$ represents the weight in $G_w$.
    On the other hand, every closed walk counted by $\hom(\vec{C}_{kb}, G)$ goes through at least one vertex of $G$ originated in $G_w$,
    and thus is also counted in $\tr(A^k)$.
    The second part follows since every closed walk is necessarily of length $kb$ for some $k\in \mathbb{N}$.
    \end{claimproof}
    Consequentially, homomorphisms over directed cycles  determine the traces of powers and vice versa.
    By Newton's identities, it holds $\tr(A^k) = \tr(B^k)$ for all $1 \leq k \leq n$ if, and only if, $\chi_A = \chi_B$.
\end{proof}

    Finally, we apply \cref{lem:uncolour} to reduce homomorphism indistinguishability over directed cycles to homomorphism indistinguishability over a $\CMSO$-definable class of undirected and uncoloured graphs of bounded pathwidth.

    \begin{proof}[Proof of \cref{thm:pw-hardness}]
        Containment follows from \cref{thm:main1}.
        Hardness follows from \cref{thm:cycles} and \cref{lem:uncolour} observing that the class of (disjoint unions) of directed cycles contains precisely all directed graphs whose vertices have out-degree $1$ and is thus definable in first-order logic.
    \end{proof}

    \begin{proof}[Proof of \cref{thm:main3}]
        \label{proof-path-hardness}
        Containment follows from \cref{lem:mwa-in-cl}.
        Hardness follows from \cref{thm:pw-hardness,thm:main1}.
    \end{proof}

    \subsection{Undirected graphs and symmetric matrices}
    \label{sec:undirected}
    
    Even though \cref{thm:pw-hardness} pinpoints the complexity of homomorphism indistinguishability over \CMSO-definable graph classes of bounded pathwidth,
    it is slightly unsatisfactory in the sense that the constructed graph class is based on a gadget construction.
    More concretely, in light of \cref{thm:cycles},
    it would be interesting to determine whether homomorphism indistinguishability over undirected cycles is \CL-complete.
    A related question\footnote{Concretely, Toda \cite{toda_counting_1984} asks whether counting the number of not necessarily simple length-$n$ paths between two vertices in an $n$-vertex undirected input graph is \GapL-complete. In the same paper, this is shown for counting directed paths in directed input graphs.} was already posed by Toda in 1984 \cite{toda_counting_1984} with little apparent progress since, see \cite{mahajan_complexity_2010}.

    In general, in the realm of logspace computation, discrepancies between problems on directed and undirected graphs
    are well known.
    For example, directed connectivity is NL\nobreakdash-complete while undirected connectivity is in \LOGSPACE \cite{reingold_undirected_2008}.
    To give another example closely related to homomorphism indistinguishability, 
    we observe invoking \cite{hoang_complexity_2000} that there is a complexity gap between similarity of symmetric and non-symmetric non-negative integer matrices unless the Exact Counting Logspace Hierarchy $\LOGSPACE^{\CL}$ collapses to \CL \cite{allender_complexity_1999,allender_relationships_1996}.
    \begin{corollary}
        \begin{enumerate}
            \item The set of pairs of similar non-negative integer matrices is $\LOGSPACE^{\CL}$\nobreakdash-complete under logspace many-one reductions.
            \item The set of pairs of similar symmetric non-negative integer matrices is in $\CL$.\footnote{In \cite{conf_version}, we erroneously claimed that this problem is $\CL$-complete. Being not able to give a proof, we leave this as an open problem.}
        \end{enumerate}
    \end{corollary}
    \begin{proof}
        The first claim follows from \cite[Theorem~4.1]{hoang_complexity_2000} and \cref{lem:pos-det}.
        The second claim follows from \cref{thm:cycles} noting that, for symmetric matrices $A$ and $B$, it holds that $\chi_A = \chi_B$ if, and only if, $A$ and $B$ are similar, see e.g.\ \cite[Theorem~3.1]{grohe_homomorphism_2025}.
    \end{proof}

    Towards illuminating the discrepancies between the directed and undirected, we show that homomorphism indistinguishability over undirected cycles and undirected cycles and paths is logspace many-one interreducible. 
    Two graphs $G$ and $H$ are homomorphism indistinguishable over cycles / cycles and paths iff there exists an orthogonal / an orthogonal pseudo-stochastic matrix $X$ such that $X A_G = A_H X$, see \cite{seppelt_homomorphism_2024}.
    The latter graph class is notable since it(s disjoint union closure) is closed under taking minors. Minor-closed graph classes play an important role in homomorphism indistinguishability \cite{roberson_oddomorphisms_2022,seppelt_logical_2024}.
    Both problems lie in $\CL$ by \cref{thm:main1}.
    \begin{theorem}\label{thm:cycles-paths}
        Homomorphism indistinguishability over the class of cycles and paths and over the class of cycles are logspace many-one interreducible.
    \end{theorem}

    For the proof of \cref{thm:cycles-paths}, we write $G + H$ and $G \times H$ for the disjoint union and the categorical product of graphs $G$ and $H$, respectively.
    By \cite[(5.28)--(5.30)]{lovasz_large_2012},
	for all graphs $F_1, F_2, G_1, G_2$, and all connected graphs $K$,
	\begin{align}
		\hom(F_1 + F_2, G) &= \hom(F_1, G) \hom(F_2, G), \label{eq:coproduct} \\
		\hom(F, G_1 \times G_2) &= \hom(F, G_1) \hom(F, G_2), \text{ and }\label{eq:product} \\
		\hom(K, G_1 + G_2) &= \hom(K, G_1) + \hom(K, G_2). \label{eq:disjoint}
	\end{align}

    \begin{proof}[Proof of \cref{thm:cycles-paths}]
        For the reduction from cycles and paths to cycles, observe that $G$ and $H$ are homomorphism indistinguishable over cycles and paths if, and only if, both $G$ and $H$ and their complements $\overline{G}$ and $\overline{H}$ are homomorphism indistinguishable over all cycles \cite[Lemma~9.4.3]{seppelt_homomorphism_2024}.
        In order to make this reduction many-one, we consider the following claim:
        \begin{claim}\label{claim:and}
            Let $\mathcal{F}$ be a graph class that is closed under taking summands.
            For graphs $G$, $G'$, $H$, and $H'$, 
            $G \equiv_{\mathcal{F}} H$ and  $G' \equiv_{\mathcal{F}} H'$
            if, and only if, $G^2 + H^2 + G'^2 + H'^2 \equiv_{\mathcal{F}} 2G \times H + 2G' \times H'$.
        \end{claim}
        \begin{claimproof}
            By \cref{eq:coproduct}, it suffices to consider homomorphism indistinguishability over connected graphs in $\mathcal{F}$, see \cite{seppelt_logical_2024}.
            For a connected graph $F$, by \cref{eq:disjoint}, it holds that $\hom(F, G + H ) = \hom(F, G) + \hom(F, H)$.
            Hence, by \cref{eq:product}, $\hom(F, G^2 + H^2 + G'^2 + H'^2) = \hom(F, 2G\times H + 2G'\times H')$
            if, and only if,
            $(\hom(F, G) - \hom(F, H))^2 + (\hom(F, G') - \hom(F, H'))^2 = 0$,
            which holds if, and only if, $\hom(F, G) = \hom(F, H)$ and $\hom(F, G') = \hom(F, H')$.
        \end{claimproof}
        Hence, the map $(G, H) \mapsto (G^2 + H^2 + \overline{G}^2 + \overline{H}^2, 2G \times H + 2\overline{G}\times \overline{H})$
        is the desired logspace many-one reduction.
        
        For the converse reduction, 
        we recall the \textsmaller{CFI} graphs \cite{cai_optimal_1992} studied in \cite{roberson_oddomorphisms_2022}.
        For a connected graph $G$, write $G^0$ and $G^1$ for its even and odd \textsmaller{CFI} graph as defined in \cite[Definition 3.1]{roberson_oddomorphisms_2022}.
        We consider the \textsmaller{CFI} graphs of the $3$-vertex cycle $C_3$ and the $4$-vertex cycle $C_4$.
        Note that $C_n^0 \cong 2C_n$ and $C_n^1 \cong C_{2n}$ for all $n \in \mathbb{N}$.

        \begin{claim}\label{claim:oddo}
            Let $F$ be a cycle or a path.
            \begin{enumerate}
                \item $\hom(F, C_3^0) \neq \hom(F, C_3^1)$ if, and only if, $F$ is an odd cycle on at least $3$ vertices.
                \item $\hom(F, C_4^0) \neq \hom(F, C_4^1)$ if, and only if, $F$ is an even cycle on at least $4$ vertices.
            \end{enumerate}
        \end{claim}
        \begin{claimproof}
            By \cite[Theorem~3.13]{roberson_oddomorphisms_2022},
            $\hom(F, C_n^0) \neq \hom(F, C_n^1)$
            if, and only if,
            $F$ admits a so-called weak oddomorphism to $C_n$.
            By \cite[Corollary~5.16]{roberson_oddomorphisms_2022},
            this implies that $F$ is not a forest, i.e.\ it is a cycle.
            The claim follows now from \cite[Corollary~3.16]{roberson_oddomorphisms_2022}.
        \end{claimproof}

        The reduction proceeds as follows.
        Given graphs $G$ and $H$,
        one may check in logspace whether $G$ and $H$ have the same number of vertices and edges, respectively.
        Now it follows from \cref{claim:oddo,eq:disjoint,eq:product}
        that $G \times C_3^0 + H \times C_3^1$ and $H \times C_3^0 + G \times C_3^1$ are homomorphism indistinguishable over all cycles and paths if, and only if, $G$ and $H$ are homomorphism indistinguishable over all odd cycles on at least $3$ vertices, and analogously for $C_4$.
        The required logspace many-one reduction can now be obtained via \cref{claim:and}.
    \end{proof}

    \section{Conclusion}

    The objective of this paper was to identify properties of graph classes $\mathcal{F}$
    that characterise the complexity of the decision problem $\HomInd(\mathcal{F})$.
    We showed that, for recognisable graph classes $\mathcal{F}$ of bounded treewidth, 
    the problem $\HomInd(\mathcal{F})$ is logspace many-one reducible to \PIT and can be \PIT-complete.
    For recognisable graph classes $\mathcal{F}$ of bounded pathwidth,
    the problem $\HomInd(\mathcal{F})$ lies in \CL and can be \CL-complete.
    In the first case, this shows optimality of the algorithm from \cite{seppelt_algorithmic_2024}
    while, in the second case, this improves upon \cite{seppelt_algorithmic_2024}.
    In the process, we show that \MWA equivalence is \CL-complete improving upon \cite{tzeng_path_1996}.

    Given the role of minor-closed graph classes in homomorphism indistinguishability \cite{roberson_oddomorphisms_2022,seppelt_logical_2024},
    it would be interesting to obtain analogous results for minor-closed graph classes of bounded treewidth and pathwidth
    (the graph classes in \cref{thm:pw-hardness,thm:recognisable-trees} are not closed under taking minors).
    
    For the first case, it was conjectured in \cite{seppelt_algorithmic_2024} that homomorphism indistinguishability is in \PTIME, which would be optimal in light of \cite{grohe_equivalence_1999}.
    Proving this, however, seems to require a better understanding of tree automata recognising minor-closed graph classes.
    The perhaps most tangible minor-closed bounded-treewidth graph class $\mathcal{F}$ for which no deterministic polynomial algorithm for $\HomInd(\mathcal{F})$ is known is the class of outerplanar graphs \cite{seppelt_algorithmic_2024}.
    In this case, $\HomInd(\mathcal{F})$ amounts to deciding exact feasibility of the first level of the Lasserre \textsmaller{SDP} hierarchy for graph isomorphism \cite{roberson_lasserre_2024}.
    
    Via \cref{thm:cycles-paths}, we have reduced the latter case to determining whether there is a complexity gap between homomorphism indistinguishability over directed and undirected cycles, which is related to an old question about \GapL-computation \cite{tzeng_path_1996,mahajan_complexity_2010}.
    Finally, in light of \cite{hutchison_testing_2006} showing that homomorphism indistinguishability over all graphs of treedepth $\leq k$ is in \TC, it is conceivable that homomorphism indistinguishability over smaller graph classes, e.g.\ of bounded treedepth, can be placed into even smaller complexity classes. 
\bibliographystyle{plainurl}

\end{document}